\documentclass[
    11pt,    %
    a4paper,    %
        ]{article}

\usepackage[
    dvipsnames   %
        ]{xcolor}

\definecolor{myOrange}{RGB}{230, 159, 0}
\definecolor{myLightBlue}{RGB}{86, 180, 233}
\definecolor{myGreen}{RGB}{0, 158, 115}
\definecolor{myYellow}{RGB}{240, 228, 66}
\definecolor{myDarkBlue}{RGB}{0, 114, 178}
\definecolor{myRed}{RGB}{213, 94, 0}
\definecolor{myPink}{RGB}{204, 121, 167}

\usepackage{soul}

\usepackage[
    T1    %
    ]{fontenc}

\usepackage{amsmath}

\usepackage{amsthm}  

\usepackage{amssymb}  

\usepackage{thmtools}
\usepackage{thm-restate}

\usepackage{framed}

\usepackage{mathtools}

\usepackage{bbold}

\usepackage{xstring}

\usepackage{tikz}

\usetikzlibrary{
    fit,    %
    calc,    %
    shapes.misc,    %
    positioning,    %
    decorations.pathreplacing    %
    }

\usepackage{subcaption}

\usepackage{rotating}

\usepackage[
    linesnumbered,    %
    boxruled,    %
    ]{algorithm2e}

\DontPrintSemicolon

\SetKwInOut{Input}{Input}
\SetKwInOut{Output}{Output}

\SetKwProg{Init}{Init:}{}{}
\SetKwProg{Proc}{Procedure}{}{end}
\SetKwProg{Insert}{Insert(\({ (x,y) }\)):}{}{}

\SetArgSty{textnormal}

\usepackage{tcolorbox}

\usepackage[
    margin=1in    %
        ]{geometry}

\makeatletter
\renewcommand{\paragraph}{%
  \@startsection{paragraph}{4}%
  {\z@}{1ex \@plus 1ex \@minus .2ex}{-1em}%
  {\normalfont\normalsize\bfseries}%
}
\makeatother

\usepackage[
    backend=bibtex,    %
    maxcitenames=10,    %
    maxbibnames=10,    %
    style=alphabetic,    %
    doi=false,    %
    isbn=false,    %
    eprint=false,    %
    url=false,    %
	backrefstyle=none,    %
        ]{biblatex}

\newbibmacro{string+doi}[1]{%
    \iffieldundef{doi}{%
        \iffieldundef{url}{#1}{\href{\thefield{url}}{#1}}
        }{%
        \href{http://dx.doi.org/\thefield{doi}}{#1}
        }
    }

\DeclareFieldFormat*{title}{%
    \usebibmacro{string+doi}{\mkbibemph{#1}}
    }

\renewbibmacro{in:}{%
    \ifentrytype{article}{}{\printtext{\bibstring{in}\intitlepunct}}}

\newtoggle{bbx:datemissing}

\renewbibmacro*{date}{\toggletrue{bbx:datemissing}%
}

\renewbibmacro*{issue+date}{}%

\newbibmacro*{date:print}{%
    \setunit{\addcomma\addspace}    %
    \togglefalse{bbx:datemissing}%
    \printdate}

\renewbibmacro*{chapter+pages}{%
    \printfield{chapter}%
    \setunit{\bibpagespunct}%
    \printfield{pages}%
    \newunit
    \usebibmacro{date:print}%
    \newunit}

\renewbibmacro*{note+pages}{%
    \printfield{note}%
    \setunit{\bibpagespunct}%
    \printfield{pages}%
    \setunit{\addcomma\addspace}%
    \usebibmacro{date}%
    \newunit}

\renewbibmacro*{addendum+pubstate}{%
    \iftoggle{bbx:datemissing}
        {\usebibmacro{date:print}}
        {}%
    \printfield{addendum}%
    \newunit\newblock
    \printfield{pubstate}}

\usepackage{xpatch}

\xpatchbibmacro{volume+number+eid}{%
  \setunit*{\adddot}%
}{%
}{}{}

\DeclareFieldFormat[article]{number}{%
    \mkbibparens{#1}
    \addcolon\space}

\appto{\bibsetup}{\sloppy}

\DefineBibliographyStrings{english}{%
    page={page},%
    pages={pages},%
    }

\DefineBibliographyStrings{english}{
  backrefpage = {cited on page},
  backrefpages = {cited on pages},
}

\usepackage[
    breaklinks=true,    %
    colorlinks=true,    %
    linkcolor=myDarkBlue,    %
    urlcolor=myRed,    %
    citecolor=myGreen,    %
    bookmarks=true,    %
    bookmarksopenlevel=2,    %
    ]{hyperref}

\usepackage[
    nameinlink    %
        ]{cleveref}

\crefname{algorithm}{algorithm}{algorithms}
\crefname{claim}{claim}{claims}
\crefname{invariant}{invariant}{invariants}
\crefname{lemma}{lemma}{lemmas}
\crefname{observation}{observation}{observations}
\crefname{problem}{problem}{problems}
\crefname{proof}{proof}{proofs}
\crefname{question}{question}{questions}

\numberwithin{equation}{section}

\theoremstyle{plain}

\newtheorem{fact}{Fact}[section]

\newtheorem{theorem}{Theorem}[section]
\newtheorem{corollary}[theorem]{Corollary}
\newtheorem{lemma}[theorem]{Lemma}
\newtheorem{proposition}[theorem]{Proposition}
\newtheorem{definition}[theorem]{Definition}

\theoremstyle{definition}

\newtheorem{question}{Question}[]

\theoremstyle{remark}
\newtheorem{remark}[theorem]{Remark}

\usepackage{xspace}

\let\originalleft\left
\let\originalright\right
\renewcommand{\left}{\mathopen{}\mathclose\bgroup\originalleft}
\renewcommand{\right}{\aftergroup\egroup\originalright}

\newcommand{\vect}[1]{\ensuremath{\boldsymbol{#1}}}
\newcommand{\poly}[1]{\ensuremath{\operatorname{poly}\left( #1 \right)}\xspace}
\newcommand{\associated}[1]{{\ensuremath{G_#1}\xspace}}
\newcommand{\stargraph}[1]{\ensuremath{G_{#1}^\star}\xspace}
\newcommand{\tree}[1]{\ensuremath{T_{#1}}\xspace}
\newcommand{\cluster}[3]{\ensuremath{C_{#1, \sigma, #2} #3}\xspace}
\newcommand{\forest}[3]{\ensuremath{F_{#1, \sigma, #2} #3}\xspace}

\newcommand{\sparsifier}[1]{\ensuremath{\tilde #1}\xspace}
\newcommand{\hypergraph}[1]{{\ensuremath{ #1}\xspace}}

\newcommand{\ConstantProp}{\ensuremath{c_{ \gamma }}\xspace}
\newcommand{\ConstantSize}{\ensuremath{c_{  }}\xspace}
\newcommand{\ConstantR}{\ensuremath{4}\xspace}
\newcommand{\ConstantVar}{\ensuremath{2}\xspace}

\newcommand{\ilast}{{\ensuremath{i_{\text{last} }}}\xspace}

\newcommand{\atmost}[1]{\ensuremath{\mathcal O\left( #1 \right)}\xspace}

\newcommand{\SpectralHypersparsifier}[1]{\ensuremath{ (1 \pm \varepsilon) }-spectral hypersparsifier\xspace}

\newcommand{\w}[1]{\ensuremath{ w_{#1} }\xspace}
\newcommand{\x}[1]{\ensuremath{ x_{#1} }\xspace}
\newcommand{\rr}[1]{\ensuremath{ r_{#1} }\xspace}

\title{Fully Dynamic Spectral Sparsification of Hypergraphs}

\usepackage{authblk}

\author[1]{Gramoz Goranci}
\author[2]{Ali Momeni}

\affil[1]{Faculty of Computer Science, University of Vienna, Austria}
\affil[2]{Faculty of Computer Science, UniVie Doctoral School Computer Science DoCS, University of Vienna, Austria}

\date{}

\begin{document}

\maketitle

\pagenumbering{gobble}    %

\begin{abstract}
Spectral hypergraph sparsification, a natural generalization of the well-studied spectral sparsification notion on graphs, has been the subject of intensive research in recent years. In this work, we consider spectral hypergraph sparsification in the dynamic setting, where the goal is to maintain a spectral sparsifier of an undirected, weighted hypergraph subject to a sequence of hyperedge insertions and deletions. For any $0 < \varepsilon \leq 1$, we give the first fully dynamic algorithm for maintaining an $ (1 \pm \varepsilon) $-spectral hypergraph sparsifier of size $ n r^3 \operatorname{poly}\left( \log n, \varepsilon ^{-1} \right) $ with amortized update time $ r^4 \operatorname{poly}\left( \log n, \varepsilon ^{-1} \right) $, where $n$ is the number of vertices of the underlying hypergraph and $r$ is an upper-bound on the rank of hyperedges. Our key contribution is to show that the spanner-based sparsification algorithm of Koutis and Xu (2016) admits a dynamic implementation in the hypergraph setting, thereby extending the dynamic spectral sparsification framework for ordinary graphs by Abraham et al. (2016).

\end{abstract}

\pagebreak

\pagenumbering{arabic}    %

\section{Introduction}

Sparsifying large graphs into smaller ones while provably retaining key graph-centric properties is a powerful algorithmic framework at the heart of scalable graph processing. A notable example is spectral sparsification, where reweighted subgraphs preserve essential spectral properties of the original graph. Rooted in the breakthrough work on ultra-fast solvers for Laplacian systems~\cite{Spielman:2004aa}, a series of advancements led to the development of essentially linear-time algorithms for computing spectral sparsifiers~\cite{Spielman:2008aa,Spielman:2011aa,Batson:2012aa}. Since then, these sparsifiers have been successfully applied to obtain significantly faster algorithms for a myriad of machine learning primitives, such as spectral clustering and graph partitioning~\cite{Lee:2014aa,Peng:2017aa,Feng:2018aa}, pruning of neural network~\cite{Hoang:2023aa,Laenen:2023aa}, and graph learning~\cite{Calandriello:2018aa}, to name a few.

In many downstream machine learning tasks, a common assumption is to endow the objects of interest with pair-wise relationships. While these relationships are naturally represented as graphs, they often fail to capture the complex interdependencies found in real-world networks, such as biological interactions or social connections~\cite{Bonacich:2004aa}. As highlighted in~\cite{Zhou:2006aa}, modeling these relationships using \emph{hypergraphs} has a host of advantages, often leading to improved performance over graph-based methods in fundamental learning tasks like clustering, classification, and embedding. This raises the following natural question -- an analogue of spectral sparsification for graphs: \emph{Do hypergraphs admit spectral sparsifiers?} Recent research has increasingly focused on this question, leading to algorithms that achieve near-optimal sparsifier sizes~\cite{Kapralov:2021aa,Kapralov:2021ab,Oko:2023aa,Jambulapati:2023aa,Lee:2023aa}.

In this work, we study spectral sparsification through the lens of an emerging and important aspect of large-scale hypergraph processing: \emph{dynamism}. Concretely, given a hypergraph that undergoes an intermixed updates of hyperedge insertions and deletions, the goal is to efficiently maintain a sparsifier that spectrally approximates the input hypergraph under these updates. Dynamic maintenance of spectral sparsifiers is particularly relevant to real-world applications, where hypergraph data is constantly changing. Recomputing solutions to hypergraph-based optimization problems from scratch after each update is often computationally prohibitive, making efficient dynamic algorithms essential.

Our main contribution is the first fully dynamic algorithm for maintaining spectral sparsifiers of hypergraphs, with strong theoretical guarantees, as summarized in the theorem below.

\begin{restatable}{theorem}{main}
\label{th:main}
Let \( 0 < \varepsilon \leq 1 \), 
\( \gamma \geq 1 \) be a constant, and \( \hypergraph{H} = (V, E, \vect{w}) \) be an initially empty \( n \)-vertex hypergraph guaranteed to have at most \( m \) hyperedges of rank 
at most \( r \) and weight ratio \( \max _{i,j} w_i / w_j \leq W \) throughout any sequence of hyperedge insertions and \atmost{n^\gamma} hyperedge deletions.
Then, there is a fully dynamic algorithm that, with probability at least \( 1 - \atmost{ \lceil \log m \rceil /n^2} \) against an oblivious adversary, maintains a \SpectralHypersparsifier{} of \( H \) with an expected size of 
\begin{equation*}
\atmost{ n r^3 \varepsilon ^{-2} \log ^2 m \log ^5 n \log W }
\end{equation*}
and an expected amortized update time of 
\begin{equation*}
\atmost{ r^4 \varepsilon ^{-2} \log ^2 m \log ^5 n \log r }.
\end{equation*}
\end{restatable}

To put our result in context, note that reading the set of vertices that define a single hyperedge of rank $r$ inherently requires $\Omega(r)$ time, implying the update time must necessarily depend on $r$. When $r = 2$, this result precisely recovers the dynamic spectral algorithm of ordinary graphs by Abraham et al.~\cite{Abraham:2016aa}. In fact, when $r$ is a big constant -- a setting  which may be relevant to many applications, our algorithm achieves both nearly optimal sparsifier size and poly-logarithmic update time in the input parameters of the problem.

\subsection{Technical Overview}

The starting point of our algorithm is the now standard sampling technique that constructs a \SpectralHypersparsifier{} \( \hypergraph{\tilde H} = (V, \tilde E, \tilde{\vect{w}}) \) of a hypergraph \( \hypergraph{H} = (V, E, \vect{w}) \) by sampling each hyperedge \( e \in E \) with probability \( p_e \) and then setting \( \tilde w_e \gets  w_e/p_e \). 
While the sampling itself can be implemented in \atmost{m} time, the challenge lies in computing \( p_e \) for each \( e \in E \) efficiently. 

The sampling probability \( p_e \) is typically inversely proportional to the effective resistance of hyperedge  \( e \) \cite{Bansal:2019aa}. As a result, if we can upper bound the effective resistance for all hyperedges, we can avoid computing \( p_e \) for every \( e \in E \).
This enables us to sample each hyperedge with constant probability \( p \) and then rescale its weight to \( \tilde w_e \gets  w_e/p \).

Motivated by \cite{Abraham:2016aa,Oko:2023aa}, which both build on the work of~\cite{Koutis:2016aa}, our algorithm uses a set of spanners \hypergraph{B} to bound the effective resistance of \( \hypergraph{H} \setminus \hypergraph{B} \), and then return \hypergraph{B} along with the sampled edges from \( \hypergraph{H} \setminus \hypergraph{B} \) (with constant probability \( p \)) as the \SpectralHypersparsifier{} of \hypergraph{H}.
More precisely, we find a sequence of \( \alpha \)-hyperspanners \( T_1, T_2, \dots, T_t \), where \( T_1 \) is an \( \alpha \)-hyperspanners of \( H \), \( T_2 \) is an \( \alpha \)-hyperspanners of \( H \setminus T_1 \), and so on.
We call \( B = T_1 \cup T_2 \cup \dots \cup T_t \) a \( t \)-bundle \( \alpha \)-hyperspanner of \( H \). 
The hypergraph $B$ guarantees that every hyperedge \( e \in H \setminus B \) has effective resistance of at most \( 4 \alpha / rt \).
By choosing a proper value for \( t \), we ensure that \( p = 1/4 \) for every hyperedge \( e \in \hypergraph{H} \setminus \hypergraph{B} \).

A priori, maintaining such a nested bundle of hyperspanners in a dynamic setting is not straightforward. To tackle this, we first design a decremental algorithm that handle only hyperedge deletions. It is not hard to see that this algorithm can be generalized to a fully dynamic one at the cost of paying a \( \log r \) overhead in the update time, where \( r \) is the rank of \hypergraph{H}.
Next, we outline the high-level idea behind our decremental algorithm.

Assume that \hypergraph{\tilde H} contains the \( t \)-bundle \( \alpha \)-hyperspanner  \( \hypergraph{B} = \hypergraph{T_1} \cup \hypergraph{T_2} \cup \dots \cup \hypergraph{T_t} \), and that a deletion of an edge from \hypergraph{H} has caused the removal of another hyperedge \( e \in T_j \) from \( B \) to maintain \hypergraph{B} as a \( t \)-bundle \( \alpha \)-hyperspanner.
Since \( T_{j+1} \) is an \( \alpha \)-spanner of \( H \setminus \cup _{i = 1} ^j T_i \), if \( e \) is not removed from \hypergraph{H}, it must appears in \( H \setminus \cup _{i = 1} ^j T_i \).
The same holds for \( T_{j+2} \) if \( e \) is not selected to be in \( T_{j+1} \) after the update, as \( T_{j+2} \) is an \( \alpha \)-spanner of \( H \setminus \cup _{i = 1} ^{j + 1} T_i \), and so on.
Thus, the deletion of \( e \) from \( T_j \) could potentially result in its insertion into the underlying hypergraphs used to maintain \( T_{j + 1}, T_{j + 2}, \dots, T_{t} \).
This is expensive for two reasons: (1) handling the insertion of \( e \) into the underlying hypergraphs requires designing a fully dynamic algorithm to maintain the \( \alpha \)-hyperspanners, and (2) the insertion would need to be passed to \atmost{t} algorithms maintaining \( T_{j + 1}, T_{j+2}, \dots, T_{t} \).

To overcome this challenge, we ensure that \( e \) is removed from \hypergraph{T_j} \textit{only if} it is removed from \hypergraph{H}.
This property is called the \textit{monotonicity property}, and we refer to an algorithm that supports this property as a \textit{monotone algorithm}.
This addresses the concerns raised in the previous paragraph: (1) since the deletion of \( e \) from \( T_j \) does not result in the insertion of \( e \) into \( \hypergraph{H} \setminus \cup_{i = 1} ^j \hypergraph{T_i}  \), we only need to design a decremental algorithm to maintain each \( T_i \), and (2) the deletion does not affect the underlying hypergraphs used to maintain \( T_{j + 1}, T_{j+2}, \dots, T_{k} \).

\subsection{Related Work}
Because of foundational importance and practical relevance of spectral-based sparsification on graphs, various variants of such sparsifiers have been studied across different computational models. Some examples include, dynamic sparsifiers~\cite{Bernstein:2022aa}, streaming sparsifiers~\cite{Kelner:2013aa,Kapralov:2017aa,Kapralov:2019aa}, distributed and parallel sparsifiers,~\cite{Koutis:2016aa}, dynamic vertex sparsifiers~\cite{Goranci:2018aa,Durfee:2019aa,Chen:2020ac,Gao:2021aa,Axiotis:2021aa,Brand:2022aa,Dong:2022aa} and distributed vertex sparsifiers~\cite{Zhu:2021aa,Forster:2021aa}.

\section{Preliminaries}

A hypergraph \( \hypergraph{H} = (V, E, \vect{w}) \) with a vertex set \( V \), a hyperedge set \( E \) and a weight vector \( \vect{w} \in \mathbb R _+ ^{E} \) defined in \( E \), is a generalization of the standard definition of graphs in which each hyperedge consists of a subset of vertices.
We define \( m = |E| \) and \( n = |V| \) and usually refer to \hypergraph{H} as an \( m \)-hyperedge \( n \)-vertex hypergraph.
For an integer \( r \geq 2 \), we say that \hypergraph{H} is of rank \( r \) if \( |e| \leq r \) for every hyperedge \( e \in E \).
We denote by \( E(H) \) the set of hyperedges of \( H \).

For a hyperedge \( e \in E \), \w{e} is the coordinate in \vect{w} corresponding to \( e \).
Similarly, for a vector \( \vect{x} \in \mathbb R ^{V} \) defined on \( V \) and a vertex \( u \in V \), \x{u} is the coordinate in \vect{x} corresponding to \( u \).

\paragraph{Spectral Hypergraph Sparsification.}
For any input vector \( \vect{x} \in \mathbb R ^{V} \), the \textit{energy} of hypergraph \( \hypergraph{H} = (V, E, \vect{w}) \) on \vect{x} is defined as
\begin{equation*}
Q _\hypergraph{H} (\vect{x}) = \sum_{e \in E} \w{e} \cdot \max_{u, v \in e} \left( \x{u} - \x{v} \right)^2.
\end{equation*}

We use \( Q _\hypergraph{H} (\vect{x}) \) to define the notion of a spectral hypergraph sparsifier for \hypergraph{H}.
For sake of brevity, we refer to such sparsifier as a \textit{spectral hypersparsifier} of \hypergraph{H}.
\begin{definition}
Let \hypergraph{H = (V, E, \vect{w})} be a hypergraph and let \( \hypergraph{\sparsifier{H}} = (V, \tilde{E} , \tilde{\vect{w}}) \) be a sub-hypergraph of \hypergraph{H}, i.e., \( \tilde{E} \subseteq E \).
For \( 0 < \varepsilon < 1 \), we say that \hypergraph{\sparsifier{H}} is a 
\SpectralHypersparsifier{} of \( H \) if, for every vector \( \vect{x} \in \mathbb R ^{V} \) defined on \( V \),
\begin{equation} \label{eq:spectral_hypergraph}
(1 - \varepsilon) Q _\hypergraph{\sparsifier{H}} (\vect{x}) \leq Q _\hypergraph{H} (\vect{x}) \leq (1 + \varepsilon) Q _\hypergraph{\sparsifier{H}} (\vect{x}).
\end{equation}
\end{definition}

Note that for graphs (i.e., hypergraphs of rank \( 2 \)), this definition is equivalent to the definition of a \( (1 \pm \varepsilon) \)-spectral sparsifier; for a graph \( G = (V, E, \vect{w}) \) and a vector \( \vect{x} \in \mathbb R ^{V} \), it follows that
\begin{equation} \label{eq:quadratic}
Q _G (\vect{x}) = \sum_{ uv \in E} \w{uv} \left( \x{u} - \x{v} \right)^2 = \vect{x}^T \mathcal L _G \vect{x},
\end{equation}
where \( \mathcal L _G
\) is the Laplacian matrix of \( G \) defined as the \( n \times n \) matrix with non-diagonal coordinate \( uv \) equal to the negated weight \( - w_{uv} \) for every edge \( uv \in E \), and diagonal coordinate \( uu \) equal to the weighted degree \( \sum _{uv \in E} w_{uv} \) of vertex \( u \).

\paragraph{Effective Resistance}
A common approach to studying the spectral properties of a graph \( G = (V, E, \vect{w}) \) is to view it as an electrical network, where each edge \( e \in E \) acts as a resistor with resistance \( \rr{e} = 1 / \w{e} \), and each node has a potential.
Using this, \Cref{eq:quadratic} is the energy of the electrical flow when the potential vector \vect{x} is applied to the nodes. 

The \textit{effective resistance \( R_G(u, v) \)} of a pair of vertices \( u, v \in V \) is defined as the potential difference between \( u \) and \( v \) when inducing one unit of current in the network.
Alternatively, it can be defined as 
\begin{equation*}
R_G(u, v) = \max_{\vect{x} \in \mathbb R ^{V}} \frac{\left( \x{u} - \x{v} \right)^2}{\vect{x}^T \mathcal L_G \vect{x}},
\end{equation*}
and thus, bounding the effective resistance \( R_G(u, v) \) bounds the energy \( Q_G(\vect{x}) \), and vice versa. 

To compute \( R_G(u, v) \) in an electrical network (i.e., the graph obtained from \( G \) by setting \( \rr{e} = 1 / \w{e} \) for every \( e \in E \)), we will utilize the well-known rules for combining the series and parallel resistors summarized in the following facts.

\begin{fact}
A series of resistors from \( u \) to \( v \) with resistances \( r_1, r_2, \dots, r_k \) can be replaced by a single resistor with endpoint \( u \) and \( v \) and an equivalent resistance of \( \sum_{i = 1} ^k r_i \).
\end{fact}

\begin{fact}
A set of parallel resistors with endpoints \( u \) and \( v \) and resistances \( r_1, r_2, \dots, r_k \) can be replaced by a single resistor with endpoints \( u \) and \( v \) and an equivalent resistance of \( 1 / \left( \sum_{i = 1} ^k 1/r_i \right) \).
\end{fact}

Since our analysis relies on computing \( R_G(\cdot, \cdot) \), in the rest of the text, we consider \( 1 / \w{e} \) as the \textit{length} of the edge \( e \).

\paragraph{Spanners and Hyperspanners.}
We first discuss the concept of spanners in graphs and then consider its generalization to hypergraphs.
For a graph \( G = (V, E, \vect{w}) \) and a pair of vertices \( u, v \in V \), let \( d_G(u , v) \) denote the length of a shortest \( (u, v) \)-path in \( G \).
As discussed before, the length of an edge \( e \in E \) is equal to \( 1/\w{e} \).
Thus,
\begin{equation*}
d_G(u, v) = \min _{\text{\( (u, v) \)-path \( P \) in \( G \)}} \left( \sum_{e \in P} \frac{1}{\w{e}} \right).
\end{equation*}
\begin{definition}
For a graph \( G = (V, E, \vect{w}) \) and \( \alpha \geq 1 \), a subgraph \( G' = (V, E', \vect{w}') \) of \( G \) is called an \( \alpha \)-spanner of \( G \) if \( d_G(u, v) \leq d_{G'}(u, v) \leq \alpha d_G(u, v) \) for every pair of vertices \( u, v \in V \).
\end{definition}

To discuss the notion of hyperspanners, we first define the concept of hyperpaths.
In a hypergraph \hypergraph{H = (V, E, \vect{w})}, a subset of hyperedges \( e_1, e_2, \dots, e_k \) is called a \( (u, v) \)-hyperpath if (1) \( u \in e_1 \), (2) \( v \in e_k \), and (3) for every \( 1 \leq i \leq k - 1 \), \( e_i \cap e_{i + 1} \neq \emptyset \).
Analogous to graphs, the length of the hyperpath is defined as \( \sum_{i = 1} ^k 1/\w{e_i} \).
We denote by \( d_\hypergraph{H} (u, v) \) the length of the shortest \( (u, v) \)-hyperpath in \hypergraph{H}.

\begin{definition}
For a hypergraph \( \hypergraph{H} = (V, E, \vect{w}) \) and \( \alpha \geq 1 \), a sub-hypergraph \( \hypergraph{H'} = (V, E', \vect{w}') \) of \hypergraph{H} is called an \( \alpha \)-hyperspanner of \hypergraph{H} if \( d_\hypergraph{H}(u, v) \leq d_\hypergraph{H'}(u, v) \leq \alpha d_\hypergraph{H}(u, v) \) for every pair of vertices \( u, v \in V \).
\end{definition}

\paragraph{The Associated Graph.}
The associated graph of a hypergraph \hypergraph{H = (V, E, \vect{w})} is a graph \associated{\hypergraph{H} = (V, E_H, \vect{w}_H)}, defined as follows: for every hyperedge \( e \in E \) in the hypergraph, we create a clique \( C(e) \) in \associated{\hypergraph{H}} on the vertices of \( e \), with all edges in \( C(e) \) assigned a weight equal to \w{e}.
We define a function \( f \colon E_H \to E \), which maps each edge \( e_H \in E_H  \) of the associated graph \associated{\hypergraph{H}} to the hyperedge \( e \in E \) where \( e_H \in C(e) \).

Note that \associated{\hypergraph{H}} may have parallel edges, as a pair of vertices can appear in several hyperedges.

\paragraph{\( t \)-Bundle \( \alpha \)-spanners and \( t \)-Bundle \( \alpha \)-Hyperspanners.}

Given a graph \( G = (V, E, \vect{w}) \), a \( t \)-bundle \( \alpha \)-spanner \( B \) of \( G \) is a union \( T_1 \cup T_2 \cup \dots \cup T_t \) of \( \alpha \)-spanners, where \( T_i \) is an \( \alpha \)-spanner of \( G \setminus \cup_{j = 1} ^{i-1} T_j \).
In other words, \( T_1 \) is an \( \alpha \)-spanner of \( G \), \( T_2 \) is an \( \alpha \)-spanner of \( G \setminus T_1 \), and so on.

The existence of a \( t \)-bundle \( \alpha \)-spanner \( B \) of \( G \) guarantees that for each edge \( uv \in E \setminus B \), there are at least \( t \) paths with length at most \( \alpha \cdot d_G(u,v) \).
We will use this property later to bound the effective resistance \( R_{G_H}(u, v) \) of the associated graph \associated{\hypergraph{H}} of a hypergraph \hypergraph{H}.

Similarly, given a hypergraph \( \hypergraph{H} = (V, E, \vect{w}) \), a \( t \)-bundle \( \alpha \)-hyperspanner \hypergraph{B} of \hypergraph{H} is a union \(  \hypergraph{T_1} \cup \hypergraph{T_2} \cup \dots \cup \hypergraph{T_t} \) of \( \alpha \)-hyperspanners, where \hypergraph{T_i} is an \( \alpha \)-hyperspanner of \( \hypergraph{H} \setminus \cup_{j = 1} ^{i-1} \hypergraph{T_j} \).
In other words, \hypergraph{T_1} is an \( \alpha \)-hyperspanner of \hypergraph{H}, \hypergraph{T_2} is a an \( \alpha \)-hyperspanner of \( \hypergraph{H} \setminus \hypergraph{T_1} \), and so on.

\paragraph{Chernoff Bound \cite{Chernoff:1952aa}.}
We will use the multiplicative version of the Chernoff bound.
Let \( X_1, X_2, \dots, X_k \) be independent random variables with values \( 0 \) or \( 1 \); i.e., for each \( 1 \leq i \leq k \), \( X_i = 1 \) with probability \( p_i \) and \( X_i = 0 \) with probability \( 1 - p_i \).
Let \( X = \sum_{i = 1} ^k X_i \), and let \( \mu = \mathbb E \left[ X \right] = \sum_{i = 1} ^k p_i \).
Then, for every \( \delta \geq 0 \),
\begin{equation} \label{eq:Chernoff}
\mathbb P \left[ X \geq (1 + \delta) \mu \right] \leq \exp\left( - \frac{\delta ^2 \mu}{2 + \delta}  \right).
\end{equation}

\section{The Algorithm}

In this section, we explain our fully dynamic algorithm for maintaining a \SpectralHypersparsifier{} of a hypergraph \( \hypergraph{H} = (V, E, \vect{w}) \).
Following \cite{Bansal:2019aa}, we describe our algorithm under the assumption that every hyperedge is of size \( (r/2, r] \).
At the end of the section, we will generalize the algorithm to hypergraphs of rank \( r \) using the lemma below.

\begin{lemma} \label{lem:r/2}
Let \( \varepsilon > 0 \) and suppose that there is an algorithm which, given any \( m \)-hyperedge \( n \)-vertex hypergraph \hypergraph{H'} with hyperedges of size \( (r/2, r ] \), computes a \SpectralHypersparsifier{} of \hypergraph{H'} with size  
\( 
 r^{\ConstantR} \varepsilon^{ - \ConstantVar } \cdot S(m, n) \)
 in \( T(m, n, r, \varepsilon ^{-1}) \) time, where \( S \) and \( T \) are  monotone non-decreasing functions.
Then, for any \( m \)-hyperedge \( n \)-vertex hypergraph \hypergraph{H} of rank \( r \), there is an algorithm for computing a \SpectralHypersparsifier{} of \hypergraph{H} with size \atmost{r^{\ConstantR} \varepsilon^{ - \ConstantVar } \cdot S(m, n)} in \atmost{T(m, n, r, \varepsilon ^{-1}) \log r} time.
\end{lemma}

The exponents of \( r \) and \( \varepsilon \) in \Cref{lem:r/2} are adjusted for our later use. 
Before proving the lemma, we first prove the decomposability of spectral hypersparsifiers below.
This will be used in both the proof of \Cref{lem:r/2} and in the proofs that follow.

\begin{lemma}[Decomposability] \label{lem:decomposability}
Let \( \varepsilon > 0 \), \( \hypergraph{H} = (V, E, \vect{w}) \) be a hypergraph, and \( E_1, E_2, \dots, E_k \) be a partition of \( E \).
For \( 1 \leq i \leq k \), define \( \hypergraph{H_i} = (V, E_i, \vect{w}_i ) \) as the sub-hypergraph of \hypergraph{H} induced by \( E_i \).
Let \hypergraph{\sparsifier{H_i}} be a \SpectralHypersparsifier{} of \hypergraph{H_i}.
Then, the hypergraph \( \hypergraph{\sparsifier{H}} = \cup_{i = 1} ^k \hypergraph{\sparsifier{H_i}} \) is a \SpectralHypersparsifier{} of \hypergraph{H}.
\end{lemma}
\begin{proof}
By \Cref{eq:spectral_hypergraph}, for every integer \( 1 \leq i \leq k \) and every vector \( \vect{x} \in \mathbb R ^n \), we have
\begin{equation*}
(1 - \varepsilon) Q_{\hypergraph{\sparsifier{H_i}}}(\vect{x}) \leq Q_{\hypergraph{H_i}}(\vect{x}) \leq (1 + \varepsilon) Q_{\hypergraph{\sparsifier{H_i}}}(\vect{x}).
\end{equation*}
Summing over all \( i \), we have
\begin{equation*}
(1 - \varepsilon) \sum_{i = 1} ^k Q_{\hypergraph{\sparsifier{H_i}}}(\vect{x}) \leq \sum_{i = 1} ^k Q_{\hypergraph{H_i}}(\vect{x})
\leq (1 + \varepsilon) \sum_{i = 1} ^k Q_{\hypergraph{\sparsifier{H_i}}}(\vect{x}).
\end{equation*}
Since \( E_1, E_2, \dots, E_k \) is a partition of \( E \), \( \sum_{i = 1} ^k Q_{\hypergraph{\sparsifier{H_i}}}(\vect{x}) = Q_{\hypergraph{\sparsifier{H}}}(\vect{x}) \) and \( \sum_{i = 1} ^k Q_{\hypergraph{H_i}}(\vect{x}) = Q_{\hypergraph{H}}(\vect{x}) \), thereby completing the proof.
\end{proof}

\begin{proof}[Proof of \Cref{lem:r/2}]
Suppose that 
For \( 1 \leq i \leq \log r \), let \hypergraph{H_i} be the sub-hypergraph containing the hyperedges of \hypergraph{H} of size \( (2^{i - 1} , 2^i] \).
i.e., \( r_i =  2^i\).
Let 
\hypergraph{\sparsifier{H_i}} be a \( (1 \pm \varepsilon ) \) spectral hypersparsifier of \hypergraph{H_i} of size \atmost{r _i ^{\ConstantR} \varepsilon ^{ - \ConstantVar } \cdot S(m, n)}, computed in time \atmost{T(m, n, \varepsilon ^{-1} )}.
We show that the hypergraph \( \hypergraph{\sparsifier{H}} = \cup _{i = 1} ^{\log r} \hypergraph{\sparsifier{H_i}} \) has the desired properties.

Since every \hypergraph{\tilde H_i} is a \SpectralHypersparsifier{} of \hypergraph{H_i}, it follows from \Cref{lem:decomposability} that \hypergraph{\sparsifier{H}} is a \SpectralHypersparsifier{} of \hypergraph{H}.
To bound the size of \hypergraph{\sparsifier{H}}, we have
\begin{equation*}
\atmost{  \sum_{i = 1} ^ {\log r} r _i ^{\ConstantR} \varepsilon ^{ - \ConstantVar } \cdot S(m, n) } 
= 
\atmost{  \sum_{i = 1} ^ {\log r} 2 ^{\ConstantR i} \varepsilon ^{ - \ConstantVar } 
 \cdot S(m, n)  }
= 
\atmost{ r^{\ConstantR} \varepsilon^{ - \ConstantVar } \cdot S(m, n) }.
\end{equation*}
The bound on time follows immediately from the construction of \hypergraph{\sparsifier{H}}.
\end{proof}

The rest of this section is divided into two parts.
In \Cref{subsec:static}, we explain our algorithm in the static setting and prove its correctness and guarantees.
In \Cref{subsec:dynamic}, we explain how to transform the static algorithm into a dynamic one.

\subsection{Static \( (1 \pm \varepsilon) \)-Spectral Hypersparsifier} \label{subsec:static}
Our static algorithm consists of two parts (\Cref{alg:light_spectral_sparsify,alg:spectral_sparsify}).
Our algorithm uses the same approach as \cite{Koutis:2016aa,Abraham:2016aa,Oko:2023aa} by using \( t \)-bundle \( \alpha \)-hyperspanners to establish a bound on the effective resistances.
In \Cref{alg:light_spectral_sparsify}, we compute a ``slightly sparser'' \SpectralHypersparsifier{} of the input hypergraph.
Then, in \Cref{alg:spectral_sparsify}, we use this procedure recursively to ``peel off'' the hypergraph until we obtain the desired sparsity.

\begin{lemma}[Adapted from \cite{Oko:2023aa}] \label{lem:effective_resistance}
Let \( \hypergraph{H} = (V, E, \vect{w}) \) be a hypergraph with hyperedges of size \( (r/2, r] \), and let \hypergraph{B} be a \( t \)-bundle \( \alpha \)-hyperspanner of \hypergraph{H}.
Then, for every hyperedge \( e \) of \( \hypergraph{H} \setminus \hypergraph{B} \) and every pair of vertices \( u,v \in e \),
\begin{equation*}
w_e \cdot R_{G_H}(u, v) \leq \frac{4 \alpha}{rt},
\end{equation*}
where \( R_{G_H}(u, v) \)  is the effective resistance of the pair \( u, v \) in the associated graph \( \associated{\hypergraph{H}} \) of \hypergraph{H}.
\end{lemma}
\begin{proof}[Proof sketch]
Let \( e \) be a hyperedge in \( \hypergraph{H} \setminus \hypergraph{B} \), let \(u, v\) be any pair of vertices from $e$, and let \( \hypergraph{B} = \hypergraph{T_1} \cup \dots \cup \hypergraph{T_t} \).
By the definition of \( \hypergraph{B} \), for each \( 1 \leq i \leq t \), there exists a hyperpath \( \pi _i \) in \hypergraph{T_i} of length at most \( \alpha \cdot 1/ w_e \) that connects \( u \) and \( v \) (recall that the length of an edge \( e \) is \( 1/w_e \)).
Thus, there are at least \( t \) hyperedge-disjoint paths connecting \( u \) and \( v \) in \hypergraph{H}.

In the associated graph \associated{\hypergraph{H}}, each path \( \pi _i \) is represented as a sequence of intersecting cliques.
Since \( |e| \in (r/2, r] \) for every \( e \in E \), there are at least \( r / 2 \) disjoint paths between each pair of vertices within the cliques, and each of these paths contain at most two edges.
Consequently, each path \( \pi _i \) represents at least \( r/2 \) edge-disjoint parallel paths of length at most \( \alpha \cdot 2/ w_e \) in \associated{\hypergraph{H}}.

We conclude that there are at least \( rt / 2 \) edge-disjoint parallel paths in \associated{\hypergraph{H}} of length at most \( \alpha \cdot 1/ w_e \) connecting \( u \) and \( v \).
Using the parallel decomposition rule of resistances to compute the effective resistance between \( u \) and \( v \), we have
\begin{equation*}
R_{G_H}(u, v) \leq \frac{\alpha \cdot 2 / w_e}{rt / 2} = \frac{4 \alpha}{rt} \cdot \frac{1}{w_e},
\end{equation*}
as desired.
\end{proof}

By using the proposition below, combined with the bound on \( R_{G_H}(u, v) \) obtained in \Cref{lem:effective_resistance}, we can determine a suitable value for the number of spanners \( t \) of a bundle as specified in \Cref{alg:light_spectral_sparsify}.

\begin{proposition}[Adapted from \cite{Bansal:2019aa}] \label{prop:constant_probability}

Let \( \varepsilon > 0 \), \( \gamma \geq 1 \) be a constant, and \( \hypergraph{H} = (V, E, \vect{w}) \) be an \( m \)-hyperedge \( n \)-vertex hypergraph with hyperedges of size \( (r/2, r] \).
Let \( \hypergraph{\sparsifier{H}} = (V, \tilde{E}, \tilde{\vect{w}} ) \) be a hypergraph formed by independently sampling each hyperedge \( e \in E \) with probability
\begin{equation*}
\min \left\{ 1, \ConstantProp r^4 \varepsilon ^{-2} \log n \cdot w_e \left( \max_{u, v \in e} R_{G_H}(u, v) \right) \right\} \leq p_e \leq 1
\end{equation*}
and assigning \( \tilde{w}_e = w_e \cdot 1/p_e  \) if \( e \in \tilde{E} \), where \ConstantProp is a constant depending on \( \gamma \), and \( R_{G_H}(u, v) \) is the effective resistance of the pair \( u, v \) in the associated graph \associated{\hypergraph{H}} of \hypergraph{H}.
Then \hypergraph{\sparsifier{H}} is a \SpectralHypersparsifier{} with probability at least \( 1 - \atmost{1/n^{\gamma + 2}} \).
\end{proposition}

\begin{algorithm}

\KwIn{\( \varepsilon > 0 \), and hypergraph \( \hypergraph{H} = (V, E, \vect{w}) \)}
\KwOut{a ``slightly sparser'' \SpectralHypersparsifier{} \( \hypergraph{\sparsifier{H}} = (V, \tilde{E}, \tilde{\vect{w}}) \) and a \( t \)-bundle \( \alpha \)-hyperspanner \hypergraph{B} of \hypergraph{H}}

\( t \gets 16 \alpha \ConstantProp r^3 \varepsilon ^{-2} \log n \)
\tcc*{\ConstantProp is from \Cref{prop:constant_probability}}

Compute a \( t \)-bundle \( \alpha \)-hyperspanner \hypergraph{B} of \hypergraph{H}

\( \hypergraph{\sparsifier{H}} \gets \hypergraph{B} \)

\For{each hyperedge \( e \) of \( \hypergraph{H} \setminus \hypergraph{B} \) 
}{
with probability \( 1/4 \), add \( e \) to \( \tilde{E} \) and set \( \tilde{w}_e \gets 4 w_e \)
}

\Return{\( (\hypergraph{\sparsifier{H}}, \hypergraph{B}) \)}

\caption{\textsc{Slight-Spectral-Sparsify(\( \hypergraph{H}, \varepsilon \))}}
\label{alg:light_spectral_sparsify}
\end{algorithm}

In the next lemma, we bound the size of the spectral hypergraph \hypergraph{\sparsifier{H}} computed by \Cref{alg:light_spectral_sparsify} and prove that the chosen value for \( t \) is sufficient.
Note that \hypergraph{\sparsifier{H}} is not sparse and therefore cannot be directly considered as the desired \SpectralHypersparsifier{} we are looking for.

\begin{lemma} \label{lem:light-spectral-sparsify}
Let \( \varepsilon > 0 \) and let \( \hypergraph{H} = (V, E, \vect{w}) \) be an \( m \)-hyperedge \( n \)-vertex hypergraph with hyperedges of size \( (r/2, r] \).
Then, with probability at least \( 1 - \atmost{1/n^2} \), \Cref{alg:light_spectral_sparsify} returns a \SpectralHypersparsifier{} \( \hypergraph{\sparsifier{H}} = (V, \tilde{E} , \tilde{\vect{w}}) \) with \( |\tilde{E}| \leq |E(\hypergraph{B})| + |E| /2 \).
\end{lemma}
\begin{proof}
We substitute the upper bound on \( R_{G_H}(u, v) \) from \Cref{lem:effective_resistance} into \Cref{prop:constant_probability}.
Consequently, we obtain
\begin{equation*}
\ConstantProp r^4 \varepsilon ^{-2} \log n \cdot w_e \left( \max_{u, v \in e} R_{G_H}(u, v) \right) \leq \frac{2 \alpha \ConstantProp r^3 \varepsilon ^{-2} \log n}{t} = \frac{1}{4}.
\end{equation*}
Thus, it follows from \Cref{prop:constant_probability} that the hypergraph \hypergraph{\sparsifier{H}} is a \SpectralHypersparsifier{}.

To prove the upper bound on \( | \tilde{E} | \), we use a standard Chernoff bound argument.
Without loss of generality, assume \( |E| > \ConstantSize n \) for a constant \( \ConstantSize > 0 \); otherwise, the algorithm simply returns a sparse hypergraph. 
For each hyperedge \( e \in H \setminus B \), 
 define the random variable \( X_e \), where \( X_e = 1 \) with probability \( 1/4 \) and  \( X_e = 0 \) with probability \( 3/4 \).
Let \( X = \sum _{e \in E} X_e \).
Clearly, \( \mu = \mathbb E[X] = |E|/4 \).
By choosing \( \delta = 1 \) in \Cref{eq:Chernoff}, we have
\begin{equation} \label{eq:whp_light_spectral_sparsify}
\mathbb P \left[X \geq \frac{|E|}{2}\right] \leq \frac{1}{e^{ |E| / 12}} = \frac{1}{e^{ \ConstantSize m^\star / 12}} < \frac{1}{e^{ \ConstantSize n / 12}}  \leq \frac{\ConstantSize}{n^2},
\end{equation}
for a sufficiently large constant \ConstantSize.
\end{proof}

We now focus on the second part of the algorithm (\Cref{alg:spectral_sparsify}).
\Cref{lem:spectral_sparsify} establishes its correctness and provides the guarantees we will use later in the dynamic setting.

\begin{algorithm}

\KwIn{\( \varepsilon > 0 \), hypergraph \( \hypergraph{H} = (V, E, \vect{w}) \), reduction parameter \( \rho \), size threshold \( m ^\star \), and constant \( \ConstantSize > 0 \)}
\KwOut{a \SpectralHypersparsifier{} \( \hypergraph{\sparsifier{H}} = (V, \tilde{E}, \tilde{\vect{w}}) \)}

\( i \gets 0 \)

\( k \gets \lceil \log \rho \rceil \)

\( \hypergraph{H_0} \gets \hypergraph{H} \)

\( \hypergraph{B_0} \gets (V, \emptyset, \vect{0}) \)

\While(
){\( i < k \) and \( |E(\hypergraph{H_i})| \geq \ConstantSize m^\star \)}{

\( (\hypergraph{\sparsifier{H_{i + 1}}}, \hypergraph{B_{i + 1}} ) \gets \textsc{Slight-Spectral-Sparsify(\( \hypergraph{H_i}, \varepsilon / (2k) \))} \)

\( \hypergraph{H_{i+1}} \gets \hypergraph{\sparsifier{H_{i+1}}} \setminus \hypergraph{B_{i+1}} \)

\( i \gets i + 1 \)
}

\( \ilast \gets i \)

\( \hypergraph{\sparsifier{H}} \gets \bigcup _{j = 1} ^ \ilast \hypergraph{B_j} \cup \hypergraph{H_\ilast} \)

\Return{(\( \hypergraph{\sparsifier{H}}
\))
}

\caption{\textsc{Spectral-Sparsify(\( \hypergraph{H}, \varepsilon \))}}
\label{alg:spectral_sparsify}
\end{algorithm}

\begin{lemma} \label{lem:spectral_sparsify}
Let \( \varepsilon > 0 \), \( \rho \) be a positive integer, \( \hypergraph{H} = (V, E, \vect{w}) \) be an \( m \)-hyperedge \( n \)-vertex hypergraph with hyperedges of size \( (r/2, r] \), and \( m^\star \geq n \) be an integer.
Then, with probability at least \( 1 - \atmost{ \lceil \log \rho \rceil /n^2} \), \Cref{alg:spectral_sparsify} returns a \SpectralHypersparsifier{} \( \hypergraph{\sparsifier{H}} = (V, \tilde{E}, \tilde{\vect{w}}) \).
Moreover, the number of iterations before the algorithm terminates is at most
\begin{equation*}
\ilast = \min\left\{ \lceil \log \rho \rceil, \lceil \log \left( m / m^\star \right) \rceil \right\},
\end{equation*}
and the size of \hypergraph{\sparsifier{H}} is 
\begin{equation} \label{eq:SizeSparsifier}
\atmost{ \sum _{j = 1} ^ \ilast |\hypergraph{B_j}| +  \ConstantSize m^\star + m / \rho }.
\end{equation}

\end{lemma}
\begin{proof}

Note that after the algorithm terminates, \ilast is the number of iterations, and we always have \( \ilast \leq k \), where \( k = \lceil \log \rho \rceil \).

We first prove that the returned hypergraph \( \hypergraph{\sparsifier{H}} = \bigcup _{j = 1} ^\ilast \hypergraph{B_j} \cup \hypergraph{H_\ilast} \) is a \SpectralHypersparsifier{} of \hypergraph{H} with probability at least \( 1 - \atmost{\ilast/n^2} \).
To do so, we use induction to show that, for every integer \( 1 \leq p \leq \ilast \), the inequality
\begin{equation} \label{eq:upper_bound}
Q_{\hypergraph{H_{\ilast - p }}}( \vect{x} ) 
\leq
\left( 1 + \varepsilon / (2k) \right)^{p} Q_{\hypergraph{I_{p}}}( \vect{x} )
\end{equation}
holds for \( \hypergraph{I_p} =  \bigcup _{j = \ilast - p + 1} ^ \ilast \hypergraph{B_j} \cup \hypergraph{H_\ilast} \) with probability at least \( (1 \pm \varepsilon /(2k) )^{p} \).
We will use this later to prove that
 \( \hypergraph{I_\ilast} =  \bigcup _{j = 1} ^ \ilast \hypergraph{B_j} \cup \hypergraph{H_\ilast} \) is a 
 \SpectralHypersparsifier{} of \( \hypergraph{H_{0}} = \hypergraph{H} \) with probability at least \( 1 - \atmost{p/n^2} \).

For \( p = 1 \), it follows from \Cref{lem:light-spectral-sparsify}.
Suppose that the claim is true for \( p = l \).
i.e., \( Q_{\hypergraph{H_{\ilast - l }}}( \vect{x} ) \leq \left( 1 + \varepsilon / (2k) \right)^{l} Q_{\hypergraph{I_{l}}}( \vect{x} ) \)
 with probability at least \( 1- \atmost{l/n^2} \).
Since \hypergraph{\sparsifier{H_{\ilast-l}}} consists of \hypergraph{B_{\ilast-l}} and \hypergraph{H_{\ilast-l}}, for every vector \( \vect{x} \in \mathbb R ^n \) we have
\begin{align*}
Q_{ \hypergraph{\sparsifier{H_{\ilast-l}}} } (\vect{x}) 
&= Q_{\hypergraph{B_{\ilast- l }}}( \vect{x} ) 
+ Q_{\hypergraph{H_{\ilast- l }}}( \vect{x} )
\\
&\leq 
 Q_{\hypergraph{B_{\ilast- l }}}( \vect{x} ) 
+ 
\left( 1 + \varepsilon / (2k) \right)^{l}
Q_{\hypergraph{I_l}}( \vect{x} )
\\
&\leq
\left( 1 + \varepsilon / (2k) \right)^{l} 
\left( 
 Q_{\hypergraph{B_{\ilast-l} }}( \vect{x} )
 +
 Q_{\hypergraph{I_l }}( \vect{x} ) 
\right)
\\
&= \left( 1 + \varepsilon / (2k) \right)^{l} Q_{\hypergraph{I_{l + 1 }}}( \vect{x} ) ,
\end{align*}
where the last line followed from the fact that \hypergraph{I_{l + 1}} consists of \hypergraph{B_{\ilast- l}} plus the hyperedges in \hypergraph{I_{l}}.
By \Cref{alg:light_spectral_sparsify}, \hypergraph{\sparsifier{H_{\ilast - l}}} is a \( \left( 1 + \varepsilon / (2k) \right) \) spectral hypersparsifier for \hypergraph{H_{\ilast - (l+ 1) }} with probability at least \( 1 - \atmost{1/n^2} \).
Thus, for every vector \( \vect{x} \in \mathbb R ^n \) we have
\begin{equation} \label{eq:spectral_proof_1}
Q_{\hypergraph{H_{\ilast - (l+ 1) }}}( \vect{x} ) 
\leq
\left( 1 + \varepsilon / (2k) \right) Q_{ \hypergraph{\sparsifier{H_{\ilast-l}}} }( \vect{x} ) 
\leq 
\left( 1 + \varepsilon / (2k) \right)^{l+ 1} Q_{\hypergraph{I_{l + 1}}}( \vect{x} ),
\end{equation}
with probability at least \( 1 - \atmost{(l + 1)/n^2} \), proving \Cref{eq:upper_bound} with the desired probabilily.

By symmetry, 
\begin{equation} \label{eq:spectral_proof_2}
\left( 1 - \varepsilon / (2k) \right)^{l+ 1} Q_{\hypergraph{I_{l + 1 }}} ( \vect{x} ) \leq Q_{\hypergraph{H_{\ilast - (l+ 1) }}} ( \vect{x} )
\end{equation}
with probability at least \( 1 - \atmost{(l + 1)/n^2} \).

It follows from \Cref{eq:spectral_proof_1,eq:spectral_proof_2} that \hypergraph{I_{l+1}} is a \( (1 \pm \varepsilon /(2k) )^{l+1} \) spectral hypersparsifier of \hypergraph{H_{\ilast-(l+1)}}
with probability at least \( 1 - \atmost{(l + 1)/n^2} \).
Since \( ( 1 + \varepsilon / (2k) )^{\ilast} \leq ( 1 + \varepsilon / (2k) )^{k} \leq  1 + \varepsilon \) and \( 1 - \varepsilon \leq ( 1 - \varepsilon / (2k) )^{k} \leq ( 1 - \varepsilon / (2k) )^{\ilast} \), it follows that \( \hypergraph{\sparsifier{H}} = \hypergraph{I_\ilast} \) is a \SpectralHypersparsifier{} of \hypergraph{H} with probability at least \( 1 - \atmost{\ilast / n^2} \).

\underline{Number of iterations:}
note that the loop halts either when \( i = k \) or \( |E(\hypergraph{H_i})| \geq \ConstantSize m^*  \).
The first condition provides the first constraint on the number of iterations.
For the latter, by \Cref{lem:light-spectral-sparsify}, \( |E(\hypergraph{H_i})| \leq |E(\hypergraph{H_{i - 1}})| /2 \) with probability at least \( 1 - \atmost{1/n^2} \) for each \( 1 \leq i \leq k \).
Thus, at each iteration, at least half of the edges are removed until fewer than \( c_2 m ^\star \) edges remain. This means that the number of iterations is at most \( \lceil \log \left( m / m^\star \right) \rceil \) with probability at least \( 1 - \atmost{k/n^2} \), as \( k \) is an upper bound on the number of iterations.

\underline{Size of \hypergraph{\sparsifier{H}}:}
we only need to discuss the size of \hypergraph{H_i}, which is either bounded by \( \ConstantSize m^\star \), or is the hypergraph \hypergraph{H_k} returned by \Cref{alg:light_spectral_sparsify} in the \( k \)th iteration.
As discussed in the previous paragraph, with probability at least \( 1 - \atmost{k/n^2} \), the size of \hypergraph{H_k} is bounded by 
\begin{equation*}
\frac{|E(\hypergraph{H_0})|}{2^k} = \frac{m}{2^{\lceil \log \rho \rceil}} = \atmost{m / \rho},
\end{equation*}
thus establishing the bound of \atmost{\ConstantSize m^\star + m/\rho}.
\end{proof}

\subsection{Dynamic \( (1 \pm \varepsilon) \)-Spectral Hypersparsifier} \label{subsec:dynamic}

We now explain the dynamic algorithm for maintaining a \SpectralHypersparsifier{}.
Similar to \cite{Abraham:2016aa}, we first design a decremental algorithm, then in \Cref{subsec:fully_dynamic}, we generalize it to a fully dynamic algorithm using \Cref{lem:turn_to_fully_dynamic}.

The decremental algorithm is the decremental version of \Cref{alg:light_spectral_sparsify,alg:spectral_sparsify} explained in \Cref{subsec:decremetal_light_spectral_sparsify,subsec:decremental_spectral_sparsify}.
Before that, we first design a decremental algorithm for maintaining \( t \)-bundle \( \alpha \)-hyperspanners in \Cref{subsec:t-bundle}, where \( \alpha = \atmost{\log n} \).

\begin{lemma} \label{lem:turn_to_fully_dynamic}
Let \( \varepsilon > 0 \) and let
\( \mathcal A \) be a decremental algorithm that, for any \( n \)-vertex hypergraph \( \hypergraph{H'} = (V, E', \vect{w}') \) with \( m' \) initial hyperedges  and the weight ratio \( W = \max _{i,j} w_i / w_j \), maintains a \SpectralHypersparsifier{} \( \hypergraph{\tilde H'} \) of \( \hypergraph{H'} \) with probability at least \( 1 - \atmost{1 / n^{2}} \) with size \( S(m', n, \varepsilon ^{-1}, W) \) in amortized update time \( T(m', n, \varepsilon ^{-1}, W) \).

Then, there is a fully dynamic algorithm that, with probability at least \( 1 - \atmost{ \lceil \log m \rceil /n^{2}} \), maintains a \SpectralHypersparsifier{} of \( \hypergraph{H} = (V, E, \vect{w}) \) of size \atmost{S(m, n, \varepsilon ^{-1}, W) \log m} in amortized update time \atmost{T(m, n, \varepsilon ^{-1}, W) \log m}, where \( m \) is an upper bound on the number of the hyperedges of \( H \) at any point.
\end{lemma}
\begin{proof}

Let \( k = \lceil \log{m} \rceil \).
The fully dynamic algorithm utilizes \( \hypergraph{\mathcal A_1}, \hypergraph{\mathcal A_2}, \dots, \hypergraph{\mathcal A_k} \) to decrementally maintain  \SpectralHypersparsifier{} on sub-hypergraphs \( \hypergraph{H_1}, \hypergraph{H_2}, \dots, \hypergraph{H_k} \) with the set of hyperedges \( \hypergraph{E_1}, \hypergraph{E_2}, \dots, \hypergraph{E_k} \), respectively.
The algorithm also maintains a counter \( t \) which tracks the number of insertions in \hypergraph{H}.
In the following, we first explain how the algorithm handles insertions and deletions, and then discuss its guarantees.

\underline{Insertions:}
upon an insertion of a hyperedge \( e \) into \hypergraph{H}, the algorithm first updates the counter \( t \) by incrementing it by one.
Let \( j \) be the highest bit in the counter \( t \) that is flipped after \( t \) is updated, or alternatively,
\begin{equation*}
j = \max \left\{ 1 \leq i \leq k \mid \text{\( t \) is divisible by \( 2^{i - 1} \)} \right\}.
\end{equation*}
The algorithm then sets \( E_j \gets \{ e \} \bigcup \cup_{i = 1} ^j E_i \) and reinitializes the decremental algorithm \( \mathcal A_i \) on \hypergraph{H_j}.
For \( 1 \leq i \leq j - 1 \), the algorithm sets \( E_i \gets \emptyset \).

\underline{Deletions:}
upon a deletion of a hyperedge \( e \) from \hypergraph{H}, the algorithm simply passes the deletion of \( e \) to the decremental algorithm \( \mathcal A_j \) on the hypergraph \hypergraph{H_j} that contains \( e \), maintaining a \SpectralHypersparsifier{} \hypergraph{\sparsifier{H_j}} of \hypergraph{H_j} after the deletion of \( e \).
Note that \( E_i \cap E_j = \emptyset \) for every \( i \neq j \).

\underline{Guarantees:}
We now prove that \( \hypergraph{\sparsifier{H}} = \cup_{i = 1} ^k \hypergraph{\sparsifier{H_i}} \), which is implicitly maintained by the algorithm, satisfies the desired guarantees.

Since \( \hypergraph{E_1}, \hypergraph{E_2}, \dots, \hypergraph{E_k} \) is a partition of \( E \), by \Cref{lem:decomposability}, \hypergraph{\sparsifier{H}} is a \SpectralHypersparsifier{}.

The bound on the size of \hypergraph{\sparsifier{H}} follows from the fact that the sparsifier \hypergraph{\tilde H} consists of the edges that appear in \( k = \lceil \log m \rceil \) sparsifiers, each of which of size \( S(m, n, W) \).

To bound the update time, note that for each \hypergraph{H_i}, the decremental algorithm \( \mathcal A_i \) is reinitialized and maintained at most \( m / 2^{i - 1} \) times as the reinitialization of \( \mathcal A_i \) happens when \( t \leq m \) is divisible by \( 2^{i - 1} \).

Note that, at any point, \( |E_i| \leq 2^{i - 1} \).
To prove this fact, we show that after setting \( E_i \gets \emptyset \), \( E_i \) will contain at most \( 2^{i - 1} \) edges before it is set to \( E_i \gets \emptyset \) again.
By the discussion above, \( E_i \) is set to \( \emptyset \) at time \( t \) only if the value of \( j \) at time \( t \), denoted \( j_t \), is greater than \( i \).
In this case, we set \( E_1, E_2, \dots, E_{j_t} \) to \( \emptyset \).
By the definition of \( j_t \), \( t = c 2^{j - 1} \) for some positive integer \( c \).
We set \( E_i \) to  \( \emptyset \) again after at most \( 2^{i - 1} \) insertions since for time \( t' = c 2^{j - 1} + 2^{i - 1} \), we have \( t' = \left( c 2^{j - i} + 1 \right) 2^{i - 1} \), and thus \( j_{t'} \geq i \).
Since there are at most \( 2^{i - 1} \) insertions, \( |E_i| \leq |E_1| + |E_2| + \dots + |E_i| \leq 2^{i - 1} \).

It follows that the total update time of \( \mathcal A_i \) after each initialization is bounded by \( |E_i| \cdot T(|E_i|, n, W) \leq 2^{i - 1} T(m, n, W) \).
Therefore, the total update time for maintaining \hypergraph{\sparsifier{H_i}} throughout the whole sequence of updates is \( \atmost{(m / 2^{i - 1}) \cdot 2^{i - 1} T(m, n, W)}  = \atmost{m \cdot T(m, n, W)} \).
Since \( k = \lceil \log m \rceil \), the total update time of the algorithm is \atmost{m \cdot T(m, n, W)  \log m}.

The guarantee on the probability simply follows from the fact that each hypersparsifier \hypergraph{\sparsifier{H_i}} is correctly maintained with probability at least \( 1 - \atmost{1/n^2} \) and that \( \hypergraph{\sparsifier{H}} = \cup_{i = 1} ^k \hypergraph{\sparsifier{H_i}} \).
\end{proof}

Recall that, in \Cref{subsec:static}, we assumed the hyperedges in \hypergraph{H} are of size \( (r/2, r ] \).
By \Cref{lem:r/2}, static algorithms on \( H \) can be generalized to hypergraphs of rank \( r \) with an \atmost{\log r} overhead in the running time.
In the lemma below, we prove that this approach continues to hold in the decremental setting.

\begin{lemma} \label{lem:r/2_dynamic}
Let \( \varepsilon > 0 \) and
suppose that there is a decremental algorithm which, given any \( n \)-vertex hypergraph \hypergraph{H'} with \( m \) initial hyperedges of size \( (r/2, r ] \), maintains a \SpectralHypersparsifier{} of \hypergraph{H'} with size \atmost{r^{\ConstantR} \varepsilon^{ - \ConstantVar } \cdot S(m, n)} in \( T(m, n, r, \varepsilon ^{-1}) \) total update time.
Then, for any \( n \)-vertex hypergraph \hypergraph{H} with \( m \) initial hyperedges of rank \( r \), there is a decremental algorithm for maintaining a \SpectralHypersparsifier{} of \hypergraph{H} with size \atmost{r^{\ConstantR} \varepsilon^{ - \ConstantVar } \cdot S(m, n)} in \atmost{T(m, n, r, \varepsilon ^{-1} ) \log r} total update time.
\end{lemma}
\begin{proof}
Similar to the proof of \Cref{lem:r/2}, we partition the hyperedges of \hypergraph{H} as follows.
For \( 1 \leq i \leq \log r \), let \hypergraph{H_i} be the sub-hypergraph containing the hyperedges of \hypergraph{H} of size \( (2^{i - 1} , 2^i] \).
i.e., \( r_i =  2^i\).
Let 
\hypergraph{\sparsifier{H_i}} be a \( (1 \pm \varepsilon ) \) spectral hypersparsifier of \hypergraph{H_i} maintained by the decremental algorithm with size \atmost{r _i ^{\ConstantR} \varepsilon ^{ - \ConstantVar } \cdot S(m_i, n)} in \atmost{T(m_i, n, \varepsilon ^{-1} )} total update time, where \( m_i \) is the number of hyperedges of \hypergraph{H_i}.

After each hyperedge deletion from \hypergraph{H}, we pass the deletion to the decremental algorithm of the corresponding sub-hypergraph containing the deleted hyperedge, thereby maintaining each \hypergraph{\tilde H_i} as a \SpectralHypersparsifier{} of \hypergraph{H_i}.

By \Cref{lem:decomposability}, the hypergraph \( \hypergraph{ \tilde H = \cup _{i = 1} ^{\log r}} \hypergraph{\sparsifier{H_i}} \) is a \SpectralHypersparsifier{} of \hypergraph{H}.
To bound the size of \hypergraph{\sparsifier{H}}, we have
\begin{equation*}
\atmost{  \sum_{i = 1} ^ {\log r} r _i ^{\ConstantR} \varepsilon ^{ - \ConstantVar } \cdot S(m_i, n) } 
= 
\atmost{  \sum_{i = 1} ^ {\log r} 2 ^{\ConstantR i} \varepsilon ^{ - \ConstantVar } 
 \cdot S(m, n)  }
= 
\atmost{ r^{\ConstantR} \varepsilon^{ - \ConstantVar } \cdot S(m, n) }.
\end{equation*}
The bound on time follows immediately from the construction of \hypergraph{\sparsifier{H}}.
\end{proof}

The rest of this section is dedicated to designing decremental implementation of \Cref{alg:light_spectral_sparsify,alg:spectral_sparsify}.
In the following paragraphs, we briefly discuss the main challenge in designing such algorithms.

Recall that the \SpectralHypersparsifier{} \hypergraph{\tilde H} of the hypergraph \hypergraph{H} returned by \Cref{alg:spectral_sparsify} consists of \( t \)-bundle \( \alpha \)-hyperspanners and some additional sampled hyperedges.
Assume that \hypergraph{\tilde H} contains the \( t \)-bundle \( \alpha \)-hyperspanner  \( \hypergraph{B} = \hypergraph{T_1} \cup \hypergraph{T_2} \cup \dots \cup \hypergraph{T_t} \), and that a deletion of an edge from \hypergraph{H} has caused the removal of another hyperedge \( e \in T_j \) from \( B \) to maintain \hypergraph{B} as a \( t \)-bundle \( \alpha \)-hyperspanner.

Since \( T_{j+1} \) is an \( \alpha \)-spanner of \( H \setminus \cup _{i = 1} ^j T_i \), if \( e \) is not removed from \hypergraph{H}, it must appears in \( H \setminus \cup _{i = 1} ^j T_i \).
The same holds for \( T_{j+2} \) if \( e \) is not selected to be in \( T_{j+1} \) after the update, as \( T_{j+2} \) is an \( \alpha \)-spanner of \( H \setminus \cup _{i = 1} ^{j + 1} T_i \), and so on.
Thus, the deletion of \( e \) from \( T_j \) could potentially result in its insertion into the underlying hypergraphs used to maintain \( T_{j + 1}, T_{j + 2}, \dots, T_{t} \).
This is expensive for two reasons: (1) handling the insertion of \( e \) into the underlying hypergraphs requires designing a fully dynamic algorithm to maintain the \( \alpha \)-spanners, and (2) the insertion would need to be passed to \atmost{t} algorithms maintaining \( T_{j + 1}, T_{j+2}, \dots, T_{t} \).

To overcome this challenge, we ensure that \( e \) is removed from \hypergraph{T_j} \textit{only if} it is removed from \hypergraph{H}.
This property is called the \textit{monotonicity property}, and we refer to an algorithm that supports this property as a \textit{monotone algorithm}.
This addresses the concerns raised in the previous paragraph: (1) since the deletion of \( e \) from \( T_j \) does not result in the insertion of \( e \) into \( \hypergraph{H} \setminus \cup_{i = 1} ^j \hypergraph{T_i}  \), we only need to design a decremental algorithm to maintain each \( T_i \), and (2) the deletion does not affect the underlying hypergraphs used to maintain \( T_{j + 1}, T_{j+2}, \dots, T_{k} \).

\subsubsection{Decremental Monotone \( t \)-Bundle \( \atmost{\log n} \)-Hyperspanners} \label{subsec:t-bundle}

To enforce the monotonicity property, we first reduce the problem of maintaining a \( t \)-bundle \( \alpha \)-hyperspanner of \hypergraph{H} to maintaining an \( \alpha \)-spanner \( G_\hypergraph{H}' \) on its associated graph \associated{\hypergraph{H}}, as shown in the following lemma.
We then apply the decremental monotone algorithm on graphs from \cite{Abraham:2016aa}, which is summarized in \Cref{lem:Abraham_decremental_t-bundle}.

\begin{lemma}[\cite{Oko:2023aa}] \label{lem:Oko_associated_graph}
Let \( \alpha \geq 1 \), \hypergraph{H = (V, E, \vect{w})} be a hypergraph, and \( \associated{\hypergraph{H}} = (V, E_\hypergraph{H}, \vect{w}_\hypergraph{H} ) \) be its associated graph with the function \( f \colon E_\hypergraph{H} \to E \) mapping each edge \( e_\hypergraph{H} \in E_\hypergraph{H}  \) to the corresponding hyperedge \( e \in E \).
If \( G_\hypergraph{H}' \) is an \( \alpha \)-spanner of \associated{\hypergraph{H}}, then the sub-hypergraph \( \hypergraph{H}' \) of \hypergraph{H} with the set of hyperedges \( E' = \{ f(e') \mid e' \in E(G_\hypergraph{H}')  \} \) is an \( \alpha \)-hyperspanner of \hypergraph{H}.
\end{lemma}

\begin{proof}
Since \( G_\hypergraph{H}' \) is an \( \alpha \)-spanner of \associated{\hypergraph{H}}, for every vertex \( u, v \in V \), there exists a \( (u, v) \)-path \( P \) in \( G_\hypergraph{H}' \) such that \( \sum_{e \in P} 1/w_e \leq \alpha d_\associated{\hypergraph{H}}(u, v) \), as \( \vect{w}_H \) matches \vect{w} on \( E_H \).
Using the map \( f \colon E_\hypergraph{H} \to E \), the path \( f(P) = \cup _{e \in P} f(e) \) is a \( (u, v) \)-path in \hypergraph{H'} with the length
\begin{equation*}
\sum_{e \in f(P)} \frac{1}{w_e}
\leq
\sum_{e \in P} \frac{1}{w_e}
\leq
\alpha d_\associated{\hypergraph{H}}(u, v)
= \alpha d_\hypergraph{H}(u, v),
\end{equation*}
where the last equation follows from the fact that \( d_\associated{\hypergraph{H}}(u, v)
= d_\hypergraph{H}(u, v) \) by the construction of the associated graph.
\end{proof}

\begin{restatable}[\cite{Abraham:2016aa}]{lemma}{abraham}\label{lem:Abraham_decremental_t-bundle}
Let \( G = (V, E, \vect{w}) \) be an \( n \)-vertex graph with \( m \) initial edges and the weight ratio \( W = \max _{i,j} w_i / w_j \), and \( t \geq 1 \) be an integer.
Then, there exists a decremental monotone algorithm against an oblivious adversary that maintains a \( t \)-bundle \atmost{\log n}-spanner \( B \) of \( G \) with an expected size of \atmost{t n \log ^3 n \log W} and an expected total update time of \atmost{t m \log ^3 n}.
\end{restatable}

\begin{remark}
While \cite{Abraham:2016aa} obtained the guarantees when \( G \) was a simple graph, we use \Cref{lem:Abraham_decremental_t-bundle} on associated graphs which are multi-graphs.
See \Cref{app:abraham} for an explanation of their algorithm and its extension to multi-graphs.
\end{remark}

Since the associated graph \associated{\hypergraph{H}} contains \atmost{m r^2} edges, directly applying \Cref{lem:Abraham_decremental_t-bundle} on \associated{\hypergraph{H}} results in a decremental monotone algorithm with an expected total update time of \atmost{t m  r^2 \log ^3 n}.
In the following lemma, we show how to reduce the expected total update time to \atmost{ t m r \log ^3 n}.

\begin{lemma} \label{lem:decremental_t-bundle}
Let \( \hypergraph{H} = (V, E, \vect{w}) \) be an \( n \)-vertex hypergraph of rank \( r \) with \( m \) initial hyperedges and weight ratio \( W = \max _{i,j} w_i / w_j \), and \( t \geq 1 \) be an integer.
Then, there exists a decremental monotone algorithm against an oblivious adversary that maintains a \( t \)-bundle \atmost{\log n}-hyperspanner \hypergraph{B} of \hypergraph{H} with an expected size of \atmost{ t n \log ^3 n \log W} and an expected total update time of \atmost{t m r  \log ^3 n}.
\end{lemma}
\begin{proof}
To achieve the desired update time,
we use a standard technique as follows.
Instead of applying the algorithm of \Cref{lem:Abraham_decremental_t-bundle} directly on the associated graph \associated{\hypergraph{H}}, we apply it on the star graph \stargraph{\hypergraph{H}} of \hypergraph{H}, defined as follows: for each hyperedge \( e \in E \), instead of adding a clique \( C(e) \) and forming the associated graph \associated{\hypergraph{H}}, replace each clique \( C(e) \) with a star centered at an arbitrarily chosen vertex of \( e \).
Note that \stargraph{\hypergraph{H}} is not unique.

By \Cref{lem:Oko_associated_graph} and the fact that \stargraph{\hypergraph{H}} is a \( 2 \)-spanner of \associated{\hypergraph{H}}, it is straightforward to see that for every \( \alpha \)-spanner \( S_\hypergraph{H}^\star \) of \stargraph{\hypergraph{H}}, the sub-hypergraph \hypergraph{H^\star} of \hypergraph{H} with hyperedges \( E^\star = \{ f(e^\star) \mid e^\star \in E(S_\hypergraph{H}^\star) \} \) is an \( 2 \alpha \)-hyperspanner of \hypergraph{H}.

Thus, without any asymptotic loss in the quality of the \atmost{\log n}-spanner, we can apply \Cref{lem:Abraham_decremental_t-bundle} on \stargraph{\hypergraph{H}}, which has \atmost{mr} edges, rather than using the associated graph \associated{\hypergraph{H}}.
The expected total update time follows accordingly.
\end{proof}

\subsubsection{Decremental Implementation of \Cref{alg:light_spectral_sparsify}} \label{subsec:decremetal_light_spectral_sparsify}

With the decremental monotone algorithm for maintaining \( t \)-bundle \( \atmost{\log n} \)-hyperspanners in place, we are now ready to explain the decremental implementation of \Cref{alg:light_spectral_sparsify}.

\begin{lemma} \label{lem:decremental_light-spectral-sparsify}
Let \( 0 < \varepsilon \leq 1 \),
 \( \gamma \geq 1 \) be a constant, and \( \hypergraph{H} = (V, E, \vect{w}) \) be an \( n \)-vertex hypergraph with \( m \) initial hyperedges of size \( (r/2, r] \).
Then, with probability at least \( 1 - \atmost{1 /n^2} \) against an oblivious adversary, the decremental implementation of \Cref{alg:light_spectral_sparsify} maintains a \SpectralHypersparsifier{} \( \hypergraph{\sparsifier{H}} = (V, \tilde{E}, \tilde{\vect{w}}) \) of expected size \( |\tilde{E}| \leq |E(B)| + |E| /2 \) in expected total update time \atmost{ t m r \log ^3 n}, where \( t = 16 \alpha \ConstantProp r^3 \varepsilon ^{-2} \log n \), over any arbitrary sequence of \atmost{n^\gamma} hyperedge deletions.
\end{lemma}
\begin{proof}
We first explain how we maintain the \( t \)-bundle \atmost{\log n}-hyperspanner \hypergraph{B} and the set of sampled hyperedges in \( \hypergraph{\sparsifier{H}} \setminus \hypergraph{B} \) after each deletion in \hypergraph{H}.

To maintain \hypergraph{B}, we use the decremental algorithm of \Cref{subsec:t-bundle}.
Specifically, after each deletion in \hypergraph{H}, we first pass the deletion to the algorithm of \Cref{lem:decremental_t-bundle} and update \hypergraph{B} accordingly.
Since this algorithm is monotone, a hyperedge is deleted from \hypergraph{B} only if it is removed from \hypergraph{H}.
Thus, \hypergraph{B} and \( \hypergraph{H} \setminus \hypergraph{B} \) remain decremental after each update.

Next, we check whether the deletion in \hypergraph{H} or the update in \hypergraph{B} resulted in hyperedge deletions in \( \hypergraph{H} \setminus \hypergraph{B} \): if a hyperedge is removed from \( \hypergraph{H} \setminus \hypergraph{B} \), we simply remove it from  \( \hypergraph{\sparsifier{H}} \setminus \hypergraph{B} \).
Recall that a hyperedge \( e \) in \( \hypergraph{\tilde H} \setminus \hypergraph{B} \) is sampled from \( \hypergraph{H} \setminus \hypergraph{B} \) with probability \( 1/4 \) and has weight \( 4 w_e \).
Since the deletion was done by an oblivious adversary, the updated \( \hypergraph{\tilde H} \setminus \hypergraph{B} \) still contains the sampled hyperedges of \( \hypergraph{H} \setminus \hypergraph{B} \), and no further action is required.

To prove that the guarantees hold with probability at least \( 1 - \atmost{1/n^2} \), first note that in the decremental setting, the probability of failure accumulates over at most \( c n^\gamma \) deletions, for a sufficiently large constant \( c \).
Thus, by using \Cref{prop:constant_probability}, we have
\begin{equation*}
\mathbb P \left[ \text{ \( \hypergraph{\sparsifier{H}} \setminus \hypergraph{B} \) is not correctly maintained } \right] \leq   \frac{c n^\gamma}{n^{\gamma + 2}} \leq  \frac{c}{n^2},
\end{equation*}
which means that \( \hypergraph{\sparsifier{H}} \) is a \SpectralHypersparsifier{} of \hypergraph{H} with probability at least \( 1 - c/n^2 \).
Moreover, by using \Cref{eq:whp_light_spectral_sparsify},
\begin{equation*}
\mathbb P \left[ \left|\hypergraph{\sparsifier{H}} \setminus \hypergraph{B} \right| \geq \frac{|E|}{2}\right] \leq \frac{c n^\gamma}{e^{ |E| / 12}} \leq \frac{c n^\gamma}{e^{ \ConstantSize n}}  \leq \frac{d}{n^2},
\end{equation*}
for a sufficiently large constant \( d \).
Therefore,  with probability at least \( 1 - \atmost{1/n^2} \), \hypergraph{\tilde H} is a \SpectralHypersparsifier{} of \hypergraph{H} with expected size \( |\tilde E| \leq |E(B)| + |E|/2 \).

As explained above, we only need to maintain \hypergraph{B} and do not need to inspect \( \hypergraph{H} \setminus \hypergraph{B} \) after each update.
Moreover, \hypergraph{B} is easily accessible since we explicitly maintain each spanner that forms \hypergraph{B}.
Thus, the update time of the algorithm is dominated by the update time of \hypergraph{B}, by \Cref{lem:decremental_t-bundle}, the expected total update time is \atmost{t m r \log ^3 n}.
\end{proof}

\subsubsection{Decremental Implementation of \Cref{alg:spectral_sparsify}} \label{subsec:decremental_spectral_sparsify}

We now explain the decremental algorithm for maintaining a \SpectralHypersparsifier{} of \hypergraph{H}.

\begin{lemma} \label{lem:decremental_spectral_sparsify}
Let \( 0 < \varepsilon \leq 1 \), \( \gamma \geq 1 \) be a constant,  \( \hypergraph{H} = (V, E, \vect{w}) \) be an \( n \)-vertex hypergraph with \( m \) initial hyperedges of size \( (r/2, r] \) and the weight ratio \( W = \max _{i,j} w_i / w_j \), and let \( 1 \leq \rho \leq m \) and \( m^\star \geq n \) be integers.
Then, there is a decremental algorithm that, with probability at least \( 1 - \atmost{ \lceil \log \rho \rceil /n^2} \) against an oblivious adversary, maintains a \SpectralHypersparsifier{} with an expected size of 
\begin{equation*}
\atmost{
\min\left\{ \lceil \log \rho \rceil, \lceil \log m / m^\star \rceil \right\}
\cdot
t n \log ^3 n \log W + \ConstantSize m^\star + m/\rho
}
\end{equation*}
and an expected total update time of 
\begin{equation*}
\atmost{
\min\left\{ \lceil \log \rho \rceil, \lceil \log m / m^\star \rceil \right\}
\cdot
t m r \log ^3 n 
}
\end{equation*}
over any arbitrary sequence of \atmost{n^\gamma} hyperedge deletions.
\end{lemma}
\begin{proof}
We implement a decremental version of \Cref{alg:spectral_sparsify} by decrementally maintaining the hypergraphs \( \hypergraph{H_0}, \hypergraph{H_1}, \dots, \hypergraph{H_k} \), \( \hypergraph{B_1}, \dots, \hypergraph{B_k} \), and \( \hypergraph{\sparsifier{H_1}}, \dots, \hypergraph{\sparsifier{H_k}} \).
For \( 1 \leq i \leq k \), we maintain \hypergraph{B_i} and \hypergraph{\sparsifier{H_i}} by using the algorithms of \Cref{lem:decremental_light-spectral-sparsify,lem:decremental_t-bundle}, respectively.
Since each \hypergraph{B_i} is maintained by a monotone algorithm, no hyperedge is inserted in \( \hypergraph{H_i} = \hypergraph{\sparsifier{H_i}} \setminus \hypergraph{B_i} \) for every \( 1 \leq i \leq k \).
Note that we can easily check whether \( |E_i| \leq \ConstantSize m^\star \) by counting the hyperedges in \hypergraph{H_i} during the initialization step.

By \Cref{lem:spectral_sparsify}, \Cref{alg:spectral_sparsify} terminates after \( \min\left\{ \lceil \log \rho \rceil, \lceil \log m / m^\star \rceil \right\} \) iterations.
This means that the number of \hypergraph{B_i}s is bounded by \( \min\left\{ \lceil \log \rho \rceil, \lceil \log m / m^\star \rceil \right\} \).
By \Cref{lem:decremental_t-bundle}, \( |\hypergraph{B_i}| = \atmost{t n  \log ^2 n \log W} \) for every \( 1 \leq i \leq k \), and the expected size of \hypergraph{\sparsifier{H}} directly follows by substituting the upper bound on the size of \hypergraph{B_i}s in \Cref{eq:SizeSparsifier}.

By \Cref{lem:decremental_t-bundle}, the expected total update time for maintaining each \hypergraph{B_i} is \atmost{tmr \log ^3 n}.
Since there are \( \min\left\{ \lceil \log \rho \rceil, \lceil \log m / m^\star \rceil \right\} \) such \hypergraph{B_i}s to be maintained, the expected total update time follows.

The guarantee on probability follows from the upper bound on the number of \hypergraph{B_i}s.
\end{proof}

Finally, we set the parameters for our algorithm.
Note that \( \alpha = \atmost{\log n} \) as the algorithm of \Cref{lem:decremental_light-spectral-sparsify} maintains \( t \)-bundle \atmost{\log n}-hyperspanners, and so \( t = \atmost{ r^3 \varepsilon ^{-2} \log ^2 n} \).
We set \( \rho =  \Theta\left( m \right) \) and \( m^\star = n = \poly{n} \), and so
\begin{equation*}
\min\left\{ \lceil \log \rho \rceil, \lceil \log m / m^\star \rceil \right\} = \lceil \log m \rceil.
\end{equation*}
We conclude our discussion on the decremental algorithm in the following corollary.

\begin{corollary} \label{cor}
Let \( 0 < \varepsilon \leq 1 \), \( \gamma \geq 1 \) be a constant, \( 1 \leq \rho \leq m \), and  \( \hypergraph{H} = (V, E, \vect{w}) \) be an \( n \)-vertex hypergraph with \( m \) initial hyperedges of size \( (r/2, r] \) and the weight ratio \( W = \max _{i,j} w_i / w_j \).
Then, there is a decremental algorithm that, with probability at least \( 1 - \atmost{ \lceil \log m \rceil /n^2} \) against an oblivious adversary, maintains a \SpectralHypersparsifier{} with an expected size of 
\begin{equation*}
\atmost{ n r^3 \varepsilon ^{-2} \log m \log ^5 n \log W }
\end{equation*}
and an expected total update time of 
\begin{equation*}
\atmost{ m r^4 \varepsilon ^{-2} \log m \log ^5 n }
\end{equation*}
over any arbitrary sequence of \atmost{n^\gamma} hyperedge deletions.
\end{corollary}

\subsubsection{The Fully Dynamic Algorithm} \label{subsec:fully_dynamic}

We extend the decremental algorithm of \Cref{cor} on hypergraphs with hyperedges of size \( (r/2, r] \) to general hypergraphs of rank \( r \), using \Cref{lem:r/2_dynamic}.
This results in an \atmost{\log r} overhead in the update time.

To extend the decremental algorithm to a fully dynamic one, we use \Cref{lem:turn_to_fully_dynamic}, which adds an \atmost{\log m} overhead to both the size and the update time.
Since \Cref{cor} holds for \atmost{n^\gamma} hyperedge deletions, we assume that the fully dynamic algorithm undergoes arbitrary hyperedge insertions but \atmost{n^\gamma} hyperedge deletions.
This ensures that each decremental algorithm in \Cref{lem:turn_to_fully_dynamic} undergoes \atmost{n^\gamma} hyperedge deletions, and so \( T\left( m, n, W \right) = \atmost{ m r^4 \varepsilon ^{-2} \log m \log ^5 n } \).
Note that, instead of restricting the number of deletions, we could reinitialize each decremental algorithm after \atmost{n^\gamma}deletions.
However, this could result in \atmost{2^n/n^\gamma} reinitializations, which adds an exponential overhead to \( T\left( m, n, W \right) \). 

The resulting guarantees are summarized below.

\main*

\section*{Acknowledgement}
We thank Edina Marica, Lara Ost, and A.\ R.\ Sricharan for their feedback on earlier versions of the manuscripts.

\clearpage
\pagenumbering{arabic}

\renewcommand{\thepage}{\thesection-\arabic{page}}

\appendix    %

\section{Decremental Monotone \( t \)-Bundle Spanners}\label{app:abraham}

We explain the algorithm of \cite{Abraham:2016aa} for maintaining a monotone \( t \)-bundle \atmost{\log n}-spanner.

\abraham*

The algorithm is explained in \Cref{subsec:app_tbundle}, which is built upon another algorithm for maintaining a decremental monotone \( \alpha \)-spanner, as explained in \Cref{subsec:app_spanner}. 

\begin{remark} \label{remark:weighted}
We explain the algorithms assuming the graphs are unweighted.
To generalize them to weighted graphs, we use the standard batching technique as follows.
Given a graph \( G = (V, E, \vect{w}) \) with a weight ratio \( W = \max _{i,j} w_i / w_j \), we partition \( G \) into subgraphs \( G_1, G_2, \dots, G_{\log_{} W + 1} \), where \( G_i \) contains the edges of \( G \) with weights in the range of \( [2^{i-1}, 2^{i} ) \).
Assuming each \( G_i \) is unweighted, we can employ the algorithm for unweighted graphs and compute an \( \alpha \)-spanner \( F_i \) of \( G_i \) of size \( S(m, n, \alpha) \).
It is easy to see that \( F = \cup _i F_i \) is a \( 2 \alpha \)-spanner of \( G \) of size \atmost{S(m, n, \alpha) \log_{} W}.
\end{remark}

\subsection{Decremental Monotone \( \alpha \)-Spanner} \label{subsec:app_spanner}
Given a decremental unweighted graph \( G = (V, E) \), we maintain an \( \alpha \)-spanner \( F \) of \( G \) with the guarantees given in the following lemma.

\begin{lemma}[\cite{Abraham:2016aa}] \label{lem:decremental_spanner}
Given \( 0 < \varepsilon \leq 1 \) and an integer \( k \geq 2 \), and an \( n \)-vertex graph \( G = (V, E) \) with \( m \) initial edges, there is a decremental monotone algorithm against an oblivious adversary that maintains a \( (2k - 1) \)-spanner \( F \) of \( G \) with an expected size of \atmost{k^2 n^{1 + 1/k} \log n } and an expected total update time of \atmost{k^2 m \log n}.
\end{lemma}

The algorithm is based on \cite{Baswana:2012aa}, and is adjusted by \cite{Abraham:2016aa} to support the monotonicity property.
The spanner is computed using clusters on \( G \), as defined below.
Each cluster consists of selected inter-cluster edges and intra-cluster edges.
The goal is to maintain the clusters and the selected edges after each edge deletion.

\begin{definition}[Clusters] \label{def:clusters}
Given an unweighted graph \( G = (V, E) \), a set \( S \subseteq V \) of centers, a random permutation \( \sigma \) of \( S \), and a positive integer \( i \), 
the cluster \( \cluster{S}{i}{[s]} \subseteq V \) of radius \( i \) centered at \( s \in S \) consists of a vertex \( v \) if 
\begin{enumerate}
\item \( d(v, s) \leq i \), and
\item if \( d(v, s) = i \), then \( \sigma(s) < \sigma(s') \) for every other center \( s' \) satisfying \( d(v, s') = i \).
\end{enumerate}
We define \cluster{S}{i}{} to be the clustering that consists of all the clusters of radius \( i \) centered at \( S \).
\end{definition}

For each cluster \cluster{S}{i}{[s]}, we define the tree \tree{s} consisting of a shortest path from each vertex  \( v \in \cluster{S}{i}{[s]} \) to the center \( s \) as follows.
The parent of \( v \) is the neighbor with the shortest distance from \( s \) (i.e., a neighbor in \cluster{S}{i}{[s]} with distance \( d(v, s) - 1 \) from \( s \)) and the smallest possible ordering w.r.t.\ \( \sigma \).
We define the forest \( \forest{S}{i}{} = \cup_{s \in S} T_s \).

The main difference between the algorithm of \cite{Baswana:2012aa} and that of \cite{Abraham:2016aa} is that \cite{Abraham:2016aa} takes into account the random ordering \( \sigma \) (as explained above) when computing the tree \tree{s}.
Using a straightforward modification of the algorithm of \citeauthor{Even:1981aa} \cite{Even:1981aa}, we can decrementally maintain the clustering \cluster{S}{i}{} and the forest \forest{S}{i}{} on the multi-graph \( G \).
The guarantee of the algorithm is summarized in the theorem below.

\begin{theorem}[\cite{Baswana:2012aa}] \label{th:intra_cluster}
Given an \( n \)-vertex graph \( G = (V, E) \) with \( m \) initial edges, a set \( S \subseteq V \) of centers, a random permutation \( \sigma \), and a positive integer \( i \), there is a decremental data structure that maintains the clustering \cluster{S}{i}{} and the forest 
 \forest{S}{i}{} 
in expected total time \atmost{ m i \log n}.
\end{theorem}

We now define the spanner \( F \) of \( G \) by selecting its edges from a sequence of clusterings, referred to as the \textit{clustering hierarchy} on \( G \).

\begin{definition}[Clustering hierarchy]
Given a graph \( G = (V, E) \), a random permutation \( \sigma \), and an integer \( k \geq 2 \), let \( S_0 \supseteq S_1 \supseteq \dots \supseteq S_k \) be subsets of \( V \), where \( S_0 = V \) and \( S_k = \emptyset \), such that \( S_{i+1} \) is obtained by sampling each vertex of \( S_i \) with probability \( n^{-1/k} \).

For every \( 1 \leq i \leq k \), we define the clustering \( C_i = \cluster{S_i}{i}{} \) on \( G \) with centers \( S_i \) and radius \( i \).
We denote by \( V_i \) the set of all vertices covered by \( C_i \).
i.e., \( V_i = \{ v \in V \mid d(v, S_i) \leq i \} \).
Moreover, we define the forest \( F_i = \forest{S}{i}{} \) on \( C_i \) as defined before.

We denote by \( C_1, C_2, \dots, C_k \) the clustering hierarchy on \( G \).
\end{definition}

\begin{definition}[Spanner \( F \)] \label{def:spanner_F}
Given a graph \( G = (V, E) \), a random permutation \( \sigma \), and an integer \( k \geq 2 \), let \( C_1, C_2, \dots, C_k \) be a clustering hierarchy on \( G \) with covered vertices \( V_1, V_2, \dots, V_k \) and forests \( F_1, F_2, \dots, F_k \). 
The spanner \( F \) consists of the following edges:
\begin{enumerate}
\item 
Intra-cluster Edges: For every \( 1 \leq i \leq k \), \( F \) consists of all the edges in the forest \( F_i \) defined on the clustering \( C_i \).

\item 
Inter-cluster Edges: For every \( 1 \leq i \leq k - 1 \) and for each \( v \in V_{i} \setminus V_{i + 1} \) belonging to the cluster \cluster{S_i}{i}{[s]} of \( C_i \) centered at \( s \), if \( v \) has some neighbors in the cluster \cluster{S_i}{i}{[s']} with \( s' \neq s \), then \( F \) consists of the edge \( (v, v') \) where \( v' \) is the neighbor of \( v \) in \cluster{S_i}{i}{[s']} with the smallest possible ordering w.r.t.\ \( \sigma \).
\end{enumerate}
After an update in \( G \), an edge in \( F \) may no longer correspond to items 1 and 2 above.
In such a case, we keep the edge in \( F \) until it is removed from \( G \), while adding new edges to \( F \) that satisfy items 1 or 2. 
\end{definition}

\begin{remark} \label{remark:inter_cluster}
\cite{Abraham:2016aa} lets the inter-cluster edge \( (v, v') \) to be arbitrarily chosen, but we choose \( v' \) to be the neighbor with the smallest possible ordering w.r.t.\ \( \sigma \).
Although this does not change the asymptotic analysis on the running time and the size of \( F \), it reduces the size of the spanner in implementation.
Recall that every selected edge remains in \( F \) until it is removed from \( G \).
By choosing \( v' \) w.r.t.\ \( \sigma \), when \( v \) changes its cluster, the edge connecting \( v \) to its new parent has already been added to \( F \) as an inter-cluster edge.
This eliminates the need to add  a new edge to \( F \) connecting \( v \) to its new parent.
\end{remark}

\begin{remark}
The choices for intra-cluster and inter-cluster edges in \( F \) are not unique when \( G \) is a multi-graph.
In this case, we have a set of parallel edges adjacent to the vertex with the least possible order.
Here, we arbitrarily choose one of these edges to be in \( F \). 
It is straightforward to see this approach does not change the guarantees compared to when \( G \) is simple.
\end{remark}

To prove the update time for maintaining the spanner \( F \), \cite{Baswana:2012aa} first computes the expected number of times a vertex changes its cluster in the clustering \( C_i \).

\begin{lemma}[\cite{Baswana:2012aa}] \label{lem:cluster_change}
For a vertex \( v \in V \), the expected number of times \( v \) changes its cluster in the clustering \( C_i \) is \atmost{i \log n}.
\end{lemma}

Since a change of cluster for \( v \) costs \atmost{\deg(v)} time to find the new inter-cluster and the intra-cluster edges adjacent to \( v \), the total expected update time for maintaining the edges of \( F \) that belong to \( C_i \) is \( \atmost{\sum _{v \in V} i \deg(v) \log n} = \atmost{i m \log n} \).

Since we support the monotonicity property, we keep more edges in \( F \) than the spanner of \cite{Baswana:2012aa}: due to the updates, \( v \) may change its parent in the forest \( F_i \) from \( v' \) to \( v'' \) as follows.
If the neighbor \( v'' \) of \( v \) becomes part of the cluster containing \( v \) and \( \sigma(v'') < \sigma(v') \), then \( v'' \) becomes the parent of \( v \) in \( F_i \), and we add \( (v, v'') \) to \( F \).
To support the monotonicity, we  keep the edge \( (v, v') \) in \( F \) while adding \( (v, v'') \) to it as the new edge connecting \( v \) to its parent in \( F_i \).
Consequently, to obtain the guarantees for the monotone algorithm maintaining \( F \), we need to be more elaborate and compute the expected number of times \( v \) changes its \textit{parent}.

\begin{lemma}[\cite{Abraham:2016aa}] \label{lem:parent_change}
For a vertex \( v \in V \), the expected number of times \( v \) changes its parent in the forest \( F_i \) is \atmost{i \log n}.
\end{lemma}

\begin{proof}

Since parallel edges have the same endpoints, the expected number of parent changes in \( G \) is upper bounded by the expected number of parent changes in the underlying simple graph of \( G \).
Therefore, we assume in this proof that \( G \) is a simple graph.

Suppose that \( G \) has gone under an arbitrary sequence of edge deletions.
Let \( V_j \) be the centers that had distance \( j \) from \( v \) at some point, and let \( (s, v') \) be the configuration in which \( v \) was assigned to the cluster centered at \( s \), \( d(v, s) = j \), and \( v' \) was the parent of \( v \) in \( F_i \).

We first prove that while keeping its distance equal \( j \) from \( S_i \), \( v \) cannot go back to the configuration \( (s, v') \) after moving to another configuration due to an update.
There are two possible cases.
\begin{enumerate}
\item 
If, after an edge deletion, \( v \) changes its configuration but stays in the same cluster.
i.e., \( v \) changes its configuration from \( (s, v') \) to \( (s, v'') \).
Since \( v \) has not changed its cluster, it follows that \( d(v', s) \leq d(v'', s) = j \) \textit{before} the update.
However, note that \( \sigma ( v' ) \leq \sigma ( v'' ) \) as \( v' \) had chosen before \( v'' \) to be the parent of \( v \).
Thus, if \( v \) chooses \( v'' \) as its parent only if \( d( v, s ) > j \) after the update.
As \( d(v, s) \) is monotonic, \( v \) cannot come back to \( (s, v') \) while keeping its distance form \( S_i \) equal to \( j \).

\item
If \( v \) changes its configuration to \( (s', v'') \), where \( s' \neq s \).
Since \( v \) keeps its distance from \( S_i \) to be equal \( j \), it follows that \( d(v, s) = d(v, s') = j \) \textit{before} the update.
By construction, \( \sigma ( s ) \leq \sigma ( s' ) \) as \( v \) had been assigned to the cluster centered at \( s \) before the update.
However, since \( v \) changed its cluster, \( d( v, s ) > j \) after the update.
As \( d(v, s) \) is monotonic, \( v \) cannot come back to \( (s, v') \) while keeping its distance form \( S_i \) equal to \( j \).
\end{enumerate}

We now bound the expected number of times \( v \) changes its parent while keeping \( d(v, S_i) = j \).
We only need to compute the expected number of configurations that appear throughout any arbitrary sequence of edge deletions.
Let \( k \) be the number of configurations \( v \) visits while \( d(v, S_i) = j \).
Note that \( k \leq n^2 \).
Since \( \sigma \) is a random permutation, each configuration has an equal probability to happen first, which is equal to \( 1/k \).
Similarly, after visiting the \( l \)th configuration, the probability of moving to an unvisited configuration is \( 1/(k - l + 1) \).
Thus, the expected number of configuration change is
\[
\sum_{l = 1} ^{k} \frac{1}{k - l + 1 } = \atmost{\log k} = \atmost{\log n},
\]
which also upper bounds the expected number of parent changes while \( d(v, S_i) = j \), as   several consecutive configurations might assign \( v \) to different clusters with the same parent.

To upper bound the number of times \( v \) changes its parent in \( F_i \), first note that once \( d(v, S_i) \) is increased, it cannot be decreased later as \( d(v, S_i) \) is monotonic.
Since in \( F_i \), \( 1 \leq d(v, S_i) \leq i \) and the expected number of times that \( v \) changes its parent while \( d(v, S_i) = j \) is \atmost{\log n}, it follows that the number of parent changes in \( F_i \) is \atmost{i \log n}.
\end{proof}

We are now ready to obtain the guarantees on the quality and the size of \( F \).

\begin{lemma}[\cite{Baswana:2012aa}] \label{lem:app_spanner}
\( F \) is a \( (2k - 1) \)-spanner of \( G \).
\end{lemma}
\begin{proof}
Suppose that \( (u, v) \) is an edge in \( G \setminus F \) and let \( j \) be the smallest integer such that \( u \) or \( v \) belongs to \( V_j \setminus V_{j+1} \).
Thus, \( u, v \in V_j \), which means that \( u \) and \( v \) are covered by the clusters \cluster{S_j}{j}{[s]} and \cluster{S_j}{j}{[s']} in \( C_j \), centered at \( s, s' \in S_j \), respectively.
Without loss of generality, suppose that \( u \in V_j \setminus V_{j+1} \).
There are two possibilities to consider:
\begin{enumerate}
\item 
If \( s = s' \), then \( u \) and \( v \) are connected in \( F \) by the path \( v \to s \to u \) which exists in the cluster \cluster{S_j}{j}{[s]} centered at \( s \).
Since the radius of the cluster is \( j \leq k - 1 \), the length of the path is at most \( 2k - 2 \).

\item
If \( s \neq s' \), then \( u \) and \( v \) are covered by different clusters in \( C_j \).
Since \( u \in \cluster{S_j}{j}{[s]} \) is connected to \( v \in \cluster{S_j}{j}{[s']} \) but \( (u, v) \) does not belong to \( F \), by item 2 of \Cref{def:spanner_F}, \( u \) is connected to a vertex \( v' \in \cluster{S_j}{j}{[s']} \) such that \( (v, v') \) belongs to \( F \).
Thus, the path \( v' \to s' \to v  \) in \( \cluster{S_j}{j}{[s']} \) concatenated with the edge \( (u, v') \) forms a path in \( F \) from \( u \) to \( v \) with a length of at most \( 1 + 2(k - 1) = 2k - 1 \).
\end{enumerate}
The statement follows accordingly.
\end{proof}

\begin{lemma}[\cite{Abraham:2016aa}] \label{lem:app_size}
The expected number of edges of \( F \) is \atmost{k^2 n^{1 + 1/k} \log n }.
\end{lemma}
\begin{proof}
Since each forest \( F_i \) has \atmost{n} edges, the total number of intra-cluster edges present in \( F \) at the beginning of the algorithm is \atmost{kn}.
As discussed in \Cref{remark:inter_cluster}, other edges to be added in \( F \) will be counted as inter-cluster edges.

To bound the number of inter-cluster edges, suppose that \( v \in V_i \setminus V_{i + 1} \).
Recall that the set \( S_{i + 1} \) of centers of the clustering \( C_{i + 1} \) is obtained by sampling each vertex of \( S_i \) with probability \( n ^{-1/n} \).
Thus, the expected number of removed centers from \( S_i \) to form \( S_{i + 1} \) is \atmost{n^{1/n}}, meaning that \( v \notin V_{i + 1} \) because \( v \) has lost its connections to every cluster in \( C_i \) after the removal of \atmost{n^{1/n}} centers.
Therefore, the expected number of inter-cluster edges in \( C_i \) touching \( v \) that are counted in \( F \) is \atmost{n^{1/n}}.
As \( C_i \) is decremental, \( v \) may change its parent, where by \Cref{remark:inter_cluster}, the new parent is connected by an edge already chosen as inter-cluster edge in the previous updates.
To support the monotonicity property, these edges remain in \( F \), and new inter-cluster edges touching \( v \) are added to \( F \), which are \atmost{n^{1/n}} as discussed before.
By \Cref{lem:parent_change}, the expected number of times \( v \) changes its parent in \( C_i \) is \atmost{i \log n}, which gives us \atmost{i n^{1/n} \log n } as the total expected number of inter-cluster edges touching \( v \) in \( C_i \) that appeared in \( F \).
Thus, the total number of inter-cluster edges in \( F \) is \( \sum _{v \in V} \sum _{i = 1} ^ k \atmost{i n^{1/n} \log n } = \atmost{k^2 n^{1 + 1/n} \log n } \).
\end{proof}

We conclude this section by proving \Cref{lem:decremental_spanner}.

\begin{proof}[Proof of \Cref{lem:decremental_spanner}]
\( F \) being a spanner and the guarantee on its size is proved by \Cref{lem:app_spanner,lem:app_size}.
By \Cref{def:spanner_F}, the edges stay in \( F \) until they are removed from \( G \), and so \( F \) supports the monotonicity property.
We prove the update time.

\underline{Update time:}
by \Cref{th:intra_cluster}, the clustering \( C_i \) and the forest \( F_i \) are maintained in \atmost{i m \log n} total update time.
By \Cref{lem:parent_change}, the number of times \( v \) changes its parent in \( C_i \) is \atmost{i \log n}.
Since a change of parent for \( v \) costs \atmost{\deg(v)} time to find the new inter-cluster  edges touching \( v \), the total update time to maintain all inter-cluster edges in \( C_i \) is \( \atmost{\sum _{v \in V} i \deg(v) \log n} = \atmost{i m \log n} \).
Therefore, the total update time for maintaining all the clusterings and forests is \( \sum_{i = 1} ^ k \atmost{m i \log n} = \atmost{m k^2 \log n} \)
.
\end{proof}

\subsection{Decremental Monotone \( t \)-Bundle \atmost{\log n}-Spanner} \label{subsec:app_tbundle}

We use the decremental monotone \( (2k -1) \)-spanner of 
\Cref{lem:decremental_spanner} where \( k = \lceil \log n \rceil \) to maintain the spanner \( T_1 \) of \( G \), spanner \( T_2 \) of \( G \setminus T_1 \) and so on.
Since \( \alpha = \atmost{\log n} \), \( B = \cup_{i = 1} ^t T_i \) is the resulting \( t \)-bundle \atmost{\log n}-spanner.

\begin{proof}[Proof of \Cref{lem:Abraham_decremental_t-bundle}]
We first assume that \( G \) is unweighted.
The fact that \( B \) is a \( t \)-bundle \atmost{\log n}-spanner follows immediately from its construction.
We prove the update time and that \( B \) supports the monotonicity property, and then generalize the algorithm to weighted graphs.

\underline{Monotonicity:}
we show that \( \cup_{i = 1} ^j T_i \) supports the monotonicity property by induction on \( j \).
For \( j = 1 \), \( T_1 \) is a monotone spanner of \( G \) by \Cref{lem:decremental_spanner}.
For \( j \geq 2 \), suppose that \( T_1 \cup T_2, \dots \cup T_j \) is a monotone spanner of \( G \).
i.e., an edge \( e \) is removed from \( T_1, T_2, \dots, T_j \) only if \( e \) is removed from \( G \), and so \( e \) will not be added to \( G \setminus \cup_{i = 1} ^{k} T_i \) for every \( 1 \leq k \leq j \).
Now, note that \( T_1 \cup T_2 \cup \dots \cup T_{j + 1} \) is a partition of edges, which means that an edge \( e \) is removed from \( \cup_{i = 1} ^{j+1} T_i \) iff \( e \) is removed from \( T_l \) for an integer \( 1 \leq l \leq j + 1 \).
If \( l \leq j \), then by assumption, \( e \) will not be added to \( G \setminus \cup _{i = 1} ^j T_i \).
If \( l = j + 1 \), then by \Cref{lem:decremental_spanner}, \( T_{j + 1} \) is a monotone spanner of \( G \setminus \cup_{i = 1} ^{j} T_i \), meaning that \( e \) will not be added to \( G \setminus \cup _{i = 1} ^j T_i \) for every \( 1 \leq k \leq j \).

\underline{Update time:}
we set \( k = \lceil \log n \rceil \) and by \Cref{lem:decremental_spanner}, each \( T_i \) has an expected size of \atmost{n \log ^3 n } and is maintained in an expected total update time of \atmost{m \log ^3 n}.
Thus, \( B \) has an expected size of \atmost{t n \log ^3 n } and is maintained in an expected total update time of \atmost{t m \log ^3 n}.

\underline{Generalization to weighted graphs:}
we use \Cref{remark:weighted}, which results in an \atmost{\log W} overhead in the size of the spanner.
Since the batching technique partitions the edges of the graph into subgraphs, the total update time will not change: \( \sum _i \atmost{t m_i \log ^3 n} = \atmost{t m \log ^3 n} \).
The monotonicity clearly follows, as the algorithm maintains a monotone spanner on each subgraph in the batch. 
\end{proof}

\clearpage
\pagenumbering{arabic}

\renewcommand{\thepage}{References-\arabic{page}}

\section*{References}    %
\printbibliography[heading=none]    %

\end{document}